\documentclass[final,3p]{elsarticle}
\usepackage[ruled,vlined]{algorithm2e}
\usepackage{enumitem,caption,amssymb,tabularx,enumitem,multirow,tikz,pgfplots,comment,amsmath,longtable,amsfonts,graphicx}
\usepackage{array}

\makeatletter
\renewcommand{\tnotetext}[1]{}
\makeatother

\journal{}

\newtheorem{theorem}{Theorem}

\newdefinition{definition}{Definition}
\newproof{proof}{Proof}
\setlistdepth{9}
\newlist{myEnumerate}{enumerate}{6}
\biboptions{sort}
\biboptions{compress}
\title{Blockchain-based Cloud Data Deduplication Scheme with Fair Incentives}
\author[rvt,focal]{Mallikarjun Reddy Dorsala\corref{cor1}}\ead{arjun753016@gmail.com}
\author[rvt]{V. N. Sastry}
\author[focal]{Sudhakar Chapram}
\cortext[cor1]{corresponding author}
\address[rvt]{Center for Mobile Banking, Institute for Development and Research in Banking Technology (IDRBT), Hyderabad, India}
\address[focal]{Department of Computer Science and Engineering, National Institute of Technology, Warangal, India}
\begin{document}
\begin{abstract}
With the rapid development of cloud computing, vast amounts of duplicated data are being uploaded to the cloud, wasting storage resources. Deduplication (dedup) is an efficient solution to save storage costs of cloud storage providers (CSPs) by storing only one copy of the uploaded data. However, cloud users do not benefit directly from dedup and may be reluctant to dedup their data. To motivate the cloud users towards dedup, CSPs offer incentives on storage fees. The problems with the existing dedup schemes are that they do not consider: (1) correctness - the incentive offered to a cloud user should be computed correctly without any prejudice. (2) fairness - the cloud user receives the file link and access rights of the uploaded data if and only if the CSP receives the storage fee. Meeting these requirements without a trusted party is non-trivial, and most of the existing dedup schemes do not apply. Another drawback is that most of the existing schemes emphasize incentives to cloud users but failed to provide a reliable incentive mechanism.

As public Blockchain networks emulate the properties of trusted parties, in this paper, we propose a new Blockchain-based dedup scheme to meet the above requirements. In our scheme, a smart contract computes the incentives on storage fee, and the fairness rules are encoded into the smart contract for facilitating fair payments between the CSPs and cloud users. We prove the correctness and fairness of the proposed scheme. We also design a new incentive mechanism and show that the scheme is individually rational and incentive compatible. Furthermore, we conduct experiments by implementing the designed smart contract on Ethereum local Blockchain network and list the transactional and financial costs of interacting with the designed smart contract. 
\end{abstract}

\maketitle
\section{Introduction}
With the advent of cloud computing, outsourcing data to remote cloud storage servers  becomes a common practice \cite{6522585}. However, most of the data being uploaded is redundant \cite{meyer2012study} and thus wasting large storage spaces.  Storage systems use deduplication (dedup) techniques to eliminate redundancy. Dedup eliminates the need to upload and store redundant copies of data by verifying whether a data already exist in its storage before each upload. If the check is valid, then the data is not uploaded, and simply the corresponding user is added as owner to the existing file. The file link is sent to the requested user upon successful verification of Proof-of-Ownership. Dedup is broadly classified into two categories: client-controlled (C-DEDU) and server-controlled (S-DEDU). In C-DEDU, the user instead of sending the entire file to a cloud storage provider (CSP), he sends only a $tag$ computed on file to check for duplication \cite{youn2015necessity}. In S-DEDU, the cloud server checks for the duplication when it receives data from the user \cite{liang2019game}. Although there are many advantages, dedup introduces interesting challenges. First, to prevent a CSP from accessing sensitive information, it is a common practice to encrypt the data before uploading. If the data is encrypted using conventional encryption techniques, then it is difficult to apply deduplication techniques. Because two identical files that are encrypted with two different keys will generate two different cipher-texts which cannot be compared for similarity. Encrypted data dedup schemes are proposed based on convergent encryption \cite{douceur2002reclaiming,li2014hybrid,bellare2013message,keelveedhi2013dupless}, secret sharing \cite{li2015secure}, proof-of-ownerships \cite{yan2016deduplication}, keyword search \cite{wen2015bdo}, and password-authenticated key exchange \cite{li2015secure}. 

Second, it is clear that the greatest beneficiary of the dedup is the CSP as it saves storage cost. It is required to motivate cloud users to opt for dedup by offering incentives/discounts on storage fee. In literature, several schemes are proposed for secure deduplication, but only a few schemes \cite{miao2015payment, liu2015secure,youn2015necessity,rabotka2016evaluation, liang2019game,wang2019blockchain} discuss the incentives in deduplication. Miao et al. \cite{miao2015payment} presents an incentive scheme which reduces the storage fee of the cloud user when opted for deduplication. Liu et al. \cite{liu2015secure} states that the incentives are necessary for attracting cloud user towards deduplication. Youn et al. \cite{youn2015necessity} necessitate the need for incentives for the first data uploader. Rabotka et al. \cite{rabotka2016evaluation} discusses the role of incentives in various secure deduplication schemes. Liang et al. \cite{liang2019game} constructs a game-theoretic model to compute the bounds on incentives a CSP can offer to cloud user for opting dedup. Recently, Wang et al. \cite{wang2019blockchain} discusses the role of Blockchain in facilitating payments between CSP and cloud users. 

In most of the dedup schemes, the incentives are computed based on the dedup rate ($n_{d}^{c}(t)$). Dedup rate is defined as the number of cloud users holding a data $d$ and opted for deduplication at cloud $c$ at time $t$. Even though if the best incentive mechanism is available, either the cloud provider has to be trusted by the cloud user for fair computation of dedup rate or a trusted party has to be recruited for computing dedup rate correctly. Another problem in existing schemes is that both the CSP and the cloud user assume a trusted party like a Bank, to facilitates payment transfers between them. However, hiring a trusted party is costly and finding an ideal trusted party which will behave honestly at all times is difficult.

The recent progress in Blockchain technology allows a public Blockchain network to emulate the properties of a trusted party. The public Blockchain network is trusted for the immutability of data it possesses, the correctness of the code (smart contract) execution in its environment and its availability.

In this paper, we propose a new Blockchain-based secure cloud storage system where no party can influence the dedup rate, and also provides fair payments. We guarantee a fair dedup rate even if the CSP is untrusted and the cloud users are rational. We employ a convergent encryption (CE) scheme for providing data privacy and a proof-of-ownership (PoP) scheme for proving ownership of duplicated data. We assume that both CE and PoP schemes as black-boxes in our model, and we solely focus on designing a Blockchain-based secure cloud storage system with a new incentive mechanism.

We summarize the contributions of this paper as follows:
\begin{enumerate}
	\item The contributions in this paper are two-fold: First, we design a new incentive mechanism, and second, we design a new Blockchain-based dedup scheme by leveraging the immutability, trust, and correctness properties of a public Blockchain network.
	\item The proposed incentive mechanism motivates cloud users to choose dedup while ensuring profits for CSP. Experimental analysis shows that the proposed incentive mechanism is individually rational and incentive compatible for both users and CSP.
	\item  As most of the existing schemes focus on secure deduplication, we propose a dedup scheme which emphasizes correctness of dedup rate and fair payments between cloud user and CSPs. We design a smart contract ($\mathcal{B}_{DEDU}$) to realize the correctness and fairness of the proposed scheme.
	\item  We implement the proposed smart contracts using solidity and execute them on a private Ethereum network which emulates the public Ethereum network. We test the proposed smart contract for the publicly available dataset and present the transactional and financial costs of interacting with the smart contract.
\end{enumerate}

\section{Related work}
In this section, we first discuss the existing literature of incentives in secure data deduplication, and then we discuss the schemes related to Blockchain-based data deduplication. 
Miao et al. \cite{miao2015payment} proposed one of the first works to have studied incentives in data deduplication. Their scheme encourages the cloud users to take part in data deduplication by offering incentives on storage fee. They show that the cloud users significantly saves the storage fee, and the CSP saves the storage costs when opted for dedup. However, they guarantee the fairness of storage fee by assuming a trusted party, which will compute the data deduplication rate. Another drawback is that their scheme fails to guarantee incentive compatibility for the CSP. 
 It is also not clear how the incentives are passed to cloud users who have already chosen deduplication when the new cloud users with the same data opt for deduplication. 
 Robatka and Mannan \cite{rabotka2016evaluation} have discussed various secure deduplication schemes and analyzed the incentive required for a CSP to opt those schemes. They show that CSPs do not have incentive or have tiny incentive to opt for deduplication when schemes like client-side encryption, DupLess \cite{bellare2013message}, probabilistic upload solutions \cite{harnik2010side} are used.  
 Liu et al. \cite{liu2015secure} abstractly discussed rewarding cloud users for adopting deduplication. However, they have not provided any concrete incentive model.
 Armknecht et al. \cite{armknecht2015transparent} proposed ClearBox, which encourages the CSPs to declare the deduplication rate to cloud users truthfully. They use an additional gateway to compute and attest the deduplication rate at the end of each round. However, the gateway can be compromised to report false dedup rate, and as it is a centralized entity, its availability cannot be guaranteed. Another drawback is that they have neither discussed a concrete incentive scheme nor discussed a payment mechanism between CSP and users.
 Youn et al. \cite{youn2015necessity} discussed disadvantages like privacy-disclosure risk for the first uploader of the data and the necessity of incentives to compensate for the first uploader disadvantages.
 Wang et al. \cite{wang2015modeling} constructed a defense scheme for the side-channel attacks in deduplication schemes. They modeled the attack by considering economic factors as a non-cooperative game between CSP and an attacker. 
 Liang et al. \cite{liang2019game} proposed a game-theoretic analysis of deduplication scheme. They constructed an incentive scheme based on a non-cooperative game between CSP and users. They also computed bounds on the incentives a CSP can offer to a cloud user for choosing dedup. Their extensive experimental analysis shows that their incentive mechanism satisfies the incentive rationality and incentive compatibility of both CSP and users. The main drawback in their scheme is that the CSP announces the storage fee based on dedup rate and its utility at specific intervals of time. A malicious CSP may report false deduplication rate and utility. According to \cite{liang2019game}, CSP adjusts the storage fee according to deduplication rate as the time passes, but these discounts at later times may lead to discrepancies in the storage fee paid by all the users who hold the same data and chosen deduplication.

Recently, some works\cite{li2018deduplication,wang2019blockchain,ming2022blockchain,huang2022blockchain} discussed deduplication using Blockchain network. The authors in \cite{li2018deduplication} use Blockchain to store tags computed on the files stored in the cloud.  After downloading data from the CSP, the cloud users can verify the tag of the downloaded data with the tag stored in Blockchain. Their objective is to provide integrity to the data stored in the cloud. Another Blockchain-based deduplication scheme is presented in \cite{wang2019blockchain}. They use a smart contract to facilitate payments between a CSP and a cloud user. However, the storage fee is computed by the CSP, and hence, a malicious CSP may report false storage fee. Another drawback is that they assume a fixed storage fee for all users irrespective of deduplication rate. They have also not discussed how their scheme realizes fair payments. Ming et al. \cite{ming2022blockchain} proposed a smart contract based deduplication scheme where a smart contract stores index of all the data stored at edge nodes of a cloud server. Whenever there is a request for data storage, the cloud searches the index stored at the smart contract and takes a decision to either store or to deduplicate the data. Huang et al. \cite{huang2022blockchain} also incorporated Blockchain to have efficient arbitration and also to distribute incentives. Similar to the proposed scheme, their scheme also achieves uniform payments by refunding the first uploader of the file whenever a new user opts for deduplication.
Recently, Li et al. \cite{li2022blockchain} and Song et al. \cite{song2023blockchain} demonstrated integrity auditing of deduplicated data using Blockchain. When compared to the above-discussed schemes, our proposed scheme not only provides confidentiality, integrity, incentives but also provides fair payments without a trusted party. 
\section{System model and problem statement}  
\subsection{System model}
Let $\mathbb{S} =\{\mathcal{C}, \mathcal{U}\}$ be the cloud storage system with a smart contract $\mathcal{B}_{DEDU}$ running on a public Blockchain network $\mathcal{BC}$ where $\mathcal{C} = \{c_{1}, c_{2},...,c_{C}\}$ is the set of CSPs, and $\mathcal{U} = \{u_{1}, u_{2},... ,u_{U}\}$ is the set of cloud users.
\par CSPs provide data storage service to cloud users. If a CSP receives a data storage request from a cloud user, it will check whether any cloud user has previously stored the data in its storage. If the check is negative, then it will ask the user to encrypt and upload the data. Otherwise, it will (1) verify the proof-of-ownership of the data and (2) issues a file link to the user to access the pre-stored data. 
\par Many cloud users may exist in the system, and some of them may request to store the same data. If they all accept the deduplication, then only one copy of that data is stored in the cloud. Let $\mathcal{D}=\{d_{1}, d_{2},..., d_{D}\}$ be the set of data that the users may wish to store in the cloud. Each $d \in \mathcal{D}$ belongs to at least one user. Let $N_{d}^{c}(t)$ represent the number of users having the same data $d$ at time $t$ therefore $N_{d}^{c}(t) \ge 1$ and $\Sigma_{d \in D} N_{d} = \mathcal{U}$.
\par  A smart contract facilitates fair payments. $\mathcal{B}_{DEDU}$ assures the users for a fair dedup rate (based on which the fee is computed) and assures the CSPs for a fair payment.
\par A public Blockchain network $\mathcal{BC}$ is maintained by a set of peers known as miners who will execute the $\mathcal{B}_{DEDU}$ according to an underlying consensus algorithm.

 As shown in Figure \ref{fig:BDEDUarchitecture}, the interactions between the entities in $\mathbb{S}$ are as follows\footnote{For the sake of simplicity, we assume single CSP and a single user in the overview. An inter CSP model is discussed later.}: 1) A CSP design and deploys a $\mathcal{B}_{DEDU}$ on $\mathcal{BC}$. 2) CSP retrieves $\mathcal{B}_{DEDU}$ address from $\mathcal{BC}$. 3) CSP announces $\mathcal{B}_{DEDU}$ address and ABI on a public platform. 4) A user sends a storage request to $\mathcal{B}_{DEDU}$. 5) $\mathcal{B}_{DEDU}$ computes and returns storage fee to be paid by the user. 6) The user sends the storage fee to $\mathcal{B}_{DEDU}$. 7) Depending upon the response from $\mathcal{B}_{DEDU}$ user sends (data or (tag and PoP)) to CSP. 8) CSP requests $\mathcal{B}_{DEDU}$ for storage fee information. 9) $\mathcal{B}_{DEDU}$ sends the storage fee information to CSP.  10) If the user sent the tag and PoP, then CSP verifies the PoP, and if it is correct, it will send a confirmation to $\mathcal{B}_{DEDU}$. If the user sends encrypted data, it will simply send a confirmation to $\mathcal{B}_{DEDU}$.  11) CSP sends the stored file link to the user. 12) The user sends the confirmation of receipt of file link to $\mathcal{B}_{DEDU}$. 13) $\mathcal{B}_{DEDU}$ sends the storage fee to CSP.
\begin{figure}[h!]
	\includegraphics[height=8cm,width=\textwidth]{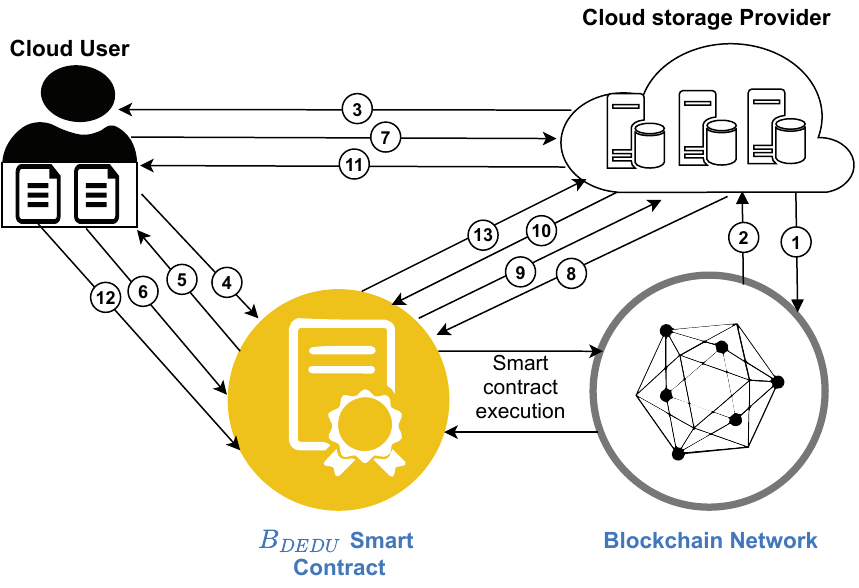}
	\caption{Overview of proposed system}
	\label{fig:BDEDUarchitecture}
\end{figure}
\subsection{Economic model}
The pricing model is pay-per-use, as often used in practice. Every party is rational and tries to maximize its utility, and the utility is based on their interactions with the proposed system.
\subsubsection{Utility of Cloud user}
Let a cloud user $u\in\mathcal{U}$ stores data $d\in\mathcal{D}$ at a cloud managed by a CSP $c\in\mathcal{C}$. Let $U_{u}^{0}(t)$ and $U_{u}^{1}(t)$ be the utilities of $u$ when he do not adopt dedup and adopt dedup with $\mathcal{B}_{DEDU}$ respectively.
\begin{equation}
U_{u}^{0}(t) = P_{u}(t) - SF_{u}^{c}(t).
\end{equation}
Where $P_{u}(t)$ is the profit earned by storing data in cloud, $SF_{u}^{c}(t)$ is the storage fee $u$ has to pay to $c$.

Now let us define the utility of a user when dedup with $\mathcal{B}_{DEDU}$ is adopted.
\begin{equation}\label{eq:Uu1}
U_{u}^{1}(t) = P_{u}(t) - \frac{SF_{u}^{c}(t)}{n^{c}_{d}(t)} - EF_{u}^{c}(t) - I_{u}(t)
\end{equation}
Where $n_{d}^{c}(t)$ is the data dedup rate at $c$, i.e., the number of users having data $d$ opted for dedup before $c$ receive $u$'s request. In our model, the fee depends on two parameters (1) Storage fee ($SF_{u}^{c}(t)$) - computed according to the current dedup rate and (2) Extra fee ($EF_{u}^{c}(t)$) - a cost paid by the user apart from storage fee to make the cloud provider incentive compatible. The cost of interacting with the smart contract is represented with $I_{u}(t)$.
\subsubsection{Utility of CSP}
Let $U_{c}^{0}(t)$ be the utility of a CSP when dedup is not applied.
\begin{equation}
U_{c}^{0}(t) = \sum_{d\in \mathcal{D}} N_{d}^{c}(t) * (SF_{u}^{c}(t)-SC_{c}^{u}(t))
\end{equation}
where $N_{d}^{c}(t)$ is the number users having data $d\in\mathcal{D}$. $SF_{u}^{c}(t)$ is the storage fee received by $c$ and $SC_{c}^{u}(t)$ is the cost incurred to $c$ for storing data.

Let $U_{c}^{1}(t)$ be the utility of CSP when dedup with $\mathcal{B}_{DEDU}$ is adopted.
\begin{multline}
U_{c}^{1}(t) = \sum_{d\in \mathcal{D}} (N^{c}_{d}(t) - n^{c}_{d}(t) + 1) * (SF_{u}^{c}(t) - SC_{c}^{u}(t))  \\ 
 + \sum_{d\in \mathcal{D}} n^{c}_{d}(t) * EF_{u}^{c}(t) 
  - \sum_{d\in \mathcal{D}} N^{c}_{d}(t) * I_{c}(t) - I_{deploy}(t) \\
\end{multline}
where $I_{c}^{deploy}(t)$ is the cost of deploying smart contract and $I_{c}(t)$ is the cost of interacting with smart contract. 

The summary of utilities of cloud user and CSP are given in Table \ref{tab:Utilitiessummary}.

\begin{table}
	\centering
	\begin{tabular}{|p{0.04\textwidth}|p{0.27\textwidth}|p{0.6\textwidth}|}
		\hline
		 & \textbf{without dedup} & \textbf{dedup with $\mathcal{B}_{DEDU}$} \\ \hline
		\textbf{User} & $U_{u}^{0}(t) = P_{u}(t) - SF_{u}^{c}(t)$ & $U_{u}^{1}(t) = P_{u}(t) -  \frac{SF_{u}^{c}(t)}{n^{c}_{d}(t)} - EF_{u}^{c}(t) - I_{u}(t)$\\ \hline
		\textbf{CSP} & $U_{c}^{0}(t) = \sum_{d\in \mathcal{D}} N_{d}^{c}(t) * (SF_{u}^{c}(t)-SC_{c}^{u}(t))$& $
		 U_{c}^{1}(t) =  \sum_{d\in \mathcal{D}} (N^{c}_{d}(t) - n^{c}_{d}(t) + 1) * (SF_{u}^{c}(t) - SC_{c}^{u}(t)) $  
		  $+ \sum_{d\in \mathcal{D}} n^{c}_{d}(t) * EF_{u}^{c}(t) $
		 $ - \sum_{d\in \mathcal{D}} N^{c}_{d}(t) * I_{c}(t) - I_{deploy}$ \\
		\hline
	\end{tabular}
\caption{Utilities of cloud user and CSP}
\label{tab:Utilitiessummary}
\end{table} 

\subsection{Problem statement}\label{def:pb}
\begin{enumerate}
	\item The first problem is to incentivize the cloud user for choosing deduplication. CSP offers this incentivization in the form of discounts on storage fee. These discounts will help to motivate the cloud user for choosing dedup with $\mathcal{B}_{DEDU}$.
	\item The second problem is to find the bounds on discount a CSP can offer so that a CSP/user earns more profits when compared to storing data without dedup.
	\item  The third problem is to compute the correct deduplication rate based on which the storage fee is computed.
	\item The fourth problem is to distribute the storage fee equally among all the users who store data $d$ at $c$ irrespective of when their request arrives.
	\item Furthermore, the last problem is to facilitate fair payments between cloud users and CSP.
\end{enumerate}
Our objective is to design a fair deduplication scheme with the following design goals.
\begin{itemize}[leftmargin=*,align=left]
\item[\textbf{Individually rational} (IR-constraint):] An incentive mechanism is said to be individually rational, if a rational user/CSP choosing deduplication with $\mathcal{B}_{DEDU}$ obtains a non-negative utility. That is $\forall u \in \mathcal{U}$ with $\mathcal{B}_{DEDU}$ $U_{u}^{1}(t) \ge 0$ and $\forall c \in \mathcal{C}$  $U_{c}^{1}(t) \ge 0$.
\item[\textbf{Incentive compatibility} (IC-constraint):] \label{def:IC-constraint}  An incentive mechanism is said to be incentive compatible, if the CSP or cloud users cannot gain more profits from not adopting dedup. The best strategy for a cloud user is to opt deduplication with $\mathcal{B}_{DEDU}$. That is $\forall u \in \mathcal{U}$ $U_{u}^{1}(t)-U_{u}^{0}(t) \ge 0$. The CSP obtains more profits by choosing deduplication with $\mathcal{B}_{DEDU}$. That is $\forall c \in \mathcal{C}$ $U_{c}^{1}(t)-U_{c}^{0}(t) \ge 0$.
\item [\textbf{Correctness}:] A dedup scheme is said to be correct if the storage fee for any user is defined by the following three factors: (1) size of the data to be stored (2) dedup rate and (3) auxiliary pricing function $p(|d|)$ at time $t$ previously set by CSP. In other words, no party (CSP/cloud users/miners) should influence the storage fee to be paid by the user except the above three factors.
\item[\textbf{Uniform payments}:] A dedup scheme is said to be uniform if every user holding data $d$ at $c$ pays the same fee irrespective of when their request arrives.
\item[\textbf{Fair payments}:] A dedup scheme is said to be fair if an honest CSP receives storage fee if and only if an honest user receives the file link of the data stored in the cloud managed by CSP.
\end{itemize}
\subsection{Threat model}
\begin{itemize}
\item \textbf{Threats to dedup rate}: Both the cloud provider and the user are greedy to increase their utility as much as possible. A malicious cloud provider or user may try to influence some of the miners of the Blockchain network to include fraudulent transactions (containing incorrect dedup rate) to increase their utility. 
\item \textbf{Threats to fairness}: Both the cloud provider and the user are allowed to abort the protocol abruptly. However, no party can gain financial advantage by aborting the protocol.
\item \textbf{Standard Blockchain model}: We assume a standard Blockchain threat model from \cite{kosba2016hawk,juels2015ring} such that the Blockchain is trusted for immutability and availability but not for privacy. We also do not consider future software updates and potential bugs of the current version of Ethereum. 
\end{itemize}
\subsection{Notations}
The notations used in this paper are presented in Table \ref{tab:notations}.
\begin{table}[!h]
	\centering
	\begin{tabular}{p{2.2cm}|p{0.85\textwidth}}
		\hline
		Symbol & Meaning  \\ \hline
		$c$ , $\mathcal{C}$ &  Cloud storage provider $c \in \mathcal{C}$, set of cloud storage providers \\
		$u$, $\mathcal{U}$ & Cloud user $u \in \mathcal{U}$, set of users  \\
		$d$, $\mathcal{D}$ & Data file $d \in \mathcal{D}$, set of data \\
		$U_{u}^{0}(t)$,  $U_{u}^{1}(t)$ & Utility of $u$ without dedup and dedup with $\mathcal{B}_{DEDU}$ at time $t$ \\
		$U_{c}^{0}(t)$, $U_{c}^{1}(t)$ & Utility of a $c$ without dedup and dedup with $\mathcal{B}_{DEDU}$ at time $t$ \\
		$p_{u}(t)$ & Profits earned by $u$ for storing data in cloud at time $t$\\
		$SF_{u}^{c}(t)$ & Storage fee paid by $u$ to $c$ at time $t$ \\
		$EF_{u}^{c}(t)$ & Extra fee paid by $u$ to $c$ apart form storage fee at time $t$\\
		$SC_{c}^{u}(t)$ & Storage cost incurred by $c$ for storing $u$'s data at time $t$\\
		$N_{d}^{c}(t)$ & Number of users having the data $d$ at $c$ at time $t$\\
		$n_{d}^{c}(t)$ & Deduplication rate of a data $d$ at cloud $c$ at time $t$\\
		$I_{deploy}(t)$ & Cost of deploying $\mathcal{B}_{DEDU}$ incurred to CSP at time $t$\\
		$I_{u}(t), I_{c}(t)$ & Cost of interacting with $\mathcal{B}_{DEDU}$ incurred to $u$ and $c$ at time $t$\\
		$\tau_{sub}$,$\tau_{p}$,$\tau_{c1}$,$\tau_{c2}$ & Timing parameters \\
		$\$d$, $\$paid$ & Monetary parameters representing deposit and fee paid by $u$ \\
		$numReq$ & Number of requests received by $\mathcal{B}_{DEDU}$  for storing file $d$\\
		$cPay$ & Current fee a user has to pay to store a file $d$\\
		$rState$ & Storage request state\\
		$ID$ & Public identity of a user which is unique across $\mathcal{BC}$\\
		$pay$ & The fee computed by $\mathcal{B}_{DEDU}$ to store a file $d$\\
		$k$ & Value on which timing parameters are computed \\
		$reqNum$ & Id of a request \\
		$uTAB$ & Special data structure maintained in the $\mathcal{B}_{DEDU}$ contract storage. $uTAB$ is of the form $uTAB($ $tag$, $numReq$, $cPay$, $requests ($ 	$ID$, $\tau_{sub}$, $\tau_{p}$, $\tau_{c1}$, $\tau_{c2}$, $rState$, $pay$, $\$paid$$))$. $uTAB$ is indexed by $tag$ and $requests$ is indexed by $reqNum$. \\
		\hline
	\end{tabular}
	\caption{Notations.}
	\label{tab:notations}
\end{table}

\section{Preliminaries}
\subsection{Convergent Encryption (CE)}
A convergent encryption scheme consists of four algorithms.
\begin{itemize}
	\item $K \leftarrow KeyGen_{CE}(d)$. The key generation algorithm takes the data file $d$ as input and outputs a convergent key $K$.
	\item $C \leftarrow Encrypt_{CE}(K,d)$. The symmetric encryption algorithm takes the convergent key $K$ and the file $d$ as input and generates a cipher-text $C$ as an output.
	\item $d \leftarrow Decrypt_{CE}(K,C)$. The decryption algorithm takes both the convergent key $K$ and cipher-text $C$ as input and outputs the original file $d$.
	\item $tag \leftarrow TG_{CE}(C)$. $TG$ is a tag generation algorithm which takes the cipher-text $C$ as input and outputs a hash value $tag$.
\end{itemize}
\subsection{Blockchain}
Bitcoin \cite{nakamoto2008bitcoin} is the first public Blockchain network proposed by Satoshi Nakamoto. Although its primary purpose is to transfer cryptocurrency known as bitcoin between peers without a trusted central authority, its fundamental concepts can be used as building blocks to construct many decentralized applications. In Blockchain network, the transactions are stored in a one-way append distributed ledger (blockchain) and a secure consensus protocol is executed by a set of decentralized peers known as miners to agree on a common global state of the ledger. The ledger holds the complete transaction history of the network. 
The blockchain is a sequence of blocks. Bitcoin uses the proof-of-work (PoW) consensus algorithm and most of the public Blockchain networks proposed later follows Bitcoin's PoW and commonly called as altcoins/Nakamoto-style ledgers.

In PoW Blockchain networks, a block $b$ is of the form $b=\langle h,t,c \rangle$ where $h \in \{0,1\}^{s}, t \in \{0,1\}^{*}, c \in \mathbb{N}$ satisfying the conditions $(G(c,H(h,t))< D)$ and $(c \leq q)$
where $G(\cdot)$ and $H(\cdot)$ are cryptographic hash functions which outputs in length of $s$ bits, $D \in \mathbb{N}$ is known as block difficulty level set by the consensus algorithm and $q \in \mathbb{N}\footnote{implementation dependent}$. Let $b^{'} = \langle h^{'},t^{'},c^{'} \rangle$, be the right most block in the chain. The chain is extended to a longer chain by adding a new block $b = \langle h,t,c \rangle$ to $b^{'}$ such that it satisfies $h=G(c^{'},H(h^{'},t^{'}))$.
A consensus algorithm guarantees the security of the blockchain. PoW algorithm makes the miners compete to generate a new block periodically. The miners are rewarded for mining new blocks in the form of currency native to the Blockchain network (bitcoins in the case of Bitcoin). The PoW algorithm has several rules out of which the following two rules ensure the correctness of the execution of transactions. (1)While generating a new block, the miners verify all the transactions going to be added in a block, and (2) Miners check the validity of a new block generated by other miners before adding it to their local blockchain.  These verification steps make the Blockchain network trusted for correctness. The hardness in solving PoW puzzle makes the blockchain immutable, and a large number of participating miners ensures availability of Blockchain network. We use these properties to achieve fair privacy-preserving aggregation without a trusted third party between data owners and data buyer.
\subsection{Smart Contract}
A smart contract is a program stored in a blockchain and executed by the mining nodes of the network. The smart contract can hold many contractual clauses between mutually distrusted parties. Similar to transactions, the smart contract is also executed by miners and, its execution correctness is guaranteed by miners executing the consensus protocol. Assuming the underlying consensus algorithm of a Blockchain is secure, the smart contract can be thought of a program executed by a trusted global machine that will faithfully execute every instruction \cite{dong2017betrayal}.

Ethereum \cite{wood2014ethereum} is a major Blockchain network supporting smart contracts. A smart contract in Ethereum is a piece of code having a contract address, balance, storage and state. 
The changes in the smart contract storage change its state.
A $\mathcal{SC}$ is a collection of functions and data similar to a class in object-oriented programming. Whenever a transaction is sent to the contract's address with a function signature, then the corresponding function code is executed by an Ethereum virtual machine (in mining node).
Unlike Bitcoin, the Ethereum scripting language is Turing-complete which motivates the developers to write smart contracts for a wide variety of applications. However, to avoid developing complex smart contracts which may take long execution and verification times jeopardizing the security of the entire Blockchain network, Ethereum introduced the concept of gas. Every opcode in Ethereum scripting language costs a pre-defined gas. Whenever a transaction makes a function call, the total gas for the function execution is computed and converted into Ether which is the native cryptocurrency of Ethereum. This Ether is charged for the transaction initiator's Ethereum account and transferred to the miner who successfully mined a new block which includes the transaction.

\section{Proposed incentive mechanism and deduplication scheme}
\subsection{Incentive mechanism}
To make a CSP incentive-compatible we introduced a new parameter $EF_{u}^{c}$ in \eqref{eq:Uu1}. In this section, we find the minimum and maximum values of $EF_{u}^{c}(t)$ so that a user/CSP obtains non-negative utility when opted for dedup with $\mathcal{B}_{DEDU}$.
\\ According to the IC-Constraint in definition \ref{def:IC-constraint}, a CSP $c$ is said to be incentive compatible if 
\begin{equation}\label{eq:cloudUtil}
U_{c}^{1}(t) - U_{c}^{0}(t) \ge 0 
\end{equation}
Substituting the utilities from Table \ref{tab:Utilitiessummary} in  \eqref{eq:cloudUtil}
\begin{multline*}
\sum_{d\in \mathcal{D}}(N_{d}^{c}(t)-n_{d}^{c}(t)+1) * (SF_{u}^{c}(t)-SC_{c}^{u}(t)) + \sum_{d\in \mathcal{D}} n_{d}^{c}(t) * EF_{u}^{c}(t) \\ - \sum_{d\in \mathcal{D}} N_{d}^{c}(t) * I_{c}(t) - I_{deploy}(t) - \sum_{d\in \mathcal{D}} N_{d}^{c}(t) * (SF_{u}^{c}(t)-SC_{c}^{u}(t) \ge 0
\end{multline*}
Assuming the cost of interacting with $\mathcal{B}_{DEDU}$ is negligible when compared to $SF_{u}^{c}(t)$ and $SC_{c}^{u}(t)$, we have
\begin{multline*}
\sum_{d\in \mathcal{D}}(N_{d}^{c}(t)-n_{d}^{c}(t)+1) * (SF_{u}^{c}(t)-SC_{c}^{u}(t)) + \sum_{d\in \mathcal{D}} n_{d}^{c}(t) * EF_{u}^{c}(t)\\ - \sum_{d\in \mathcal{D}} (N_{d}^{c}(t)) * (SF_{u}^{c}(t)-SC_{c}^{u}(t)) \ge 0
\end{multline*}
For a single data file we have
\begin{equation*}
(1-n_{d}^{c}(t)) * (SF_{u}^{c}(t)-SC_{c}^{u}(t)) + n_{d}^{c}(t) * EF_{u}^{c}(t) \ge 0
\end{equation*}
\begin{equation*}
n_{d}^{c}(t) * EF_{u}^{c}(t) \ge (n_{d}^{c}(t)-1) * (SF_{u}^{c}(t)-SC_{c}^{u}(t))
\end{equation*}
\begin{equation} \label{eq:minCF}
EF_{u}^{c}(t) \ge \frac{n_{d}^{c}(t)-1}{n_{d}^{c}(t)} * (SF_{u}^{c}(t)-SC_{c}^{u}(t))
\end{equation}
Now, we find the maximum value of $EF_{u}^{c}(t)$ so that a user is incentive-compatible when opted for dedup with $\mathcal{B}_{DEDU}$. According to the IC-Constraint in definition \ref{def:IC-constraint}, a user $u$ is said to be incentive-compatible if 
\begin{equation}\label{eq:userUtil}
U_{u}^{1}(t) - U_{u}^{0}(t) \ge 0
\end{equation}
Substituting the utilities from Table \ref{tab:Utilitiessummary} in  \eqref{eq:userUtil}
\begin{equation*}
P_{u}(t) - \frac{SF_{u}^{c}(t)}{n_{d}^{c}(t)} - EF_{u}^{c}(t) - I_{u}(t) - P_{u}(t) + SF_{u}^{c}(t) \ge 0
\end{equation*}
Assuming the cost of interacting with $\mathcal{B}_{DEDU}$ is negligible when compared to $SF_{u}^{c}(t)$, we have
\begin{equation*}
SF_{u}^{c}(t) - \frac{SF_{u}^{c}(t)}{n_{d}^{c}(t)} - EF_{u}^{c}(t) \ge 0
\end{equation*}
\begin{equation}\label{eq:maxEF}
EF_{u}^{c}(t) \le  \frac{n_{d}^{c}(t)-1}{n_{d}^{c}(t)} * SF_{u}^{c}(t)
\end{equation}
From \eqref{eq:minCF} and \eqref{eq:maxEF} the minimum and maximum values of $EF_{u}^{c}(t)$ are set as 
\begin{equation}
 EF_{u}^{c}(t) =\Bigg[ \frac{n_{d}^{c}(t)-1}{n_{d}^{c}(t)} * (SF_{u}^{c}(t)-SC_{c}^{u}(t)), \frac{n_{d}^{c}(t)-1}{n_{d}^{c}(t)} * SF_{u}^{c}(t) \Bigg] 
\end{equation}
when the cost of interacting with $\mathcal{B}_{DEDU}$ is considered then
\begin{multline*}
EF_{u}^{c}(t) =\Bigg[ \frac{n_{d}^{c}(t)-1}{n_{d}^{c}(t)} * (SF_{u}^{c}(t)-SC_{c}^{u}(t)) + (n_{d}^{c}(t) * I_{c}(t)) + I_{deploy}, \\(\frac{n_{d}^{c}(t)-1}{n_{d}^{c}(t)} * SF_{u}^{c}(t)) + I_{u}(t) \Bigg] 
\end{multline*}
\subsection{Blockchain-based deduplication scheme}
In this section, we discuss a Blockchain-based cloud storage system which consists of a smart contract $\mathcal{B}_{DEDU}$, a protocol to interact with $\mathcal{B}_{DEDU}$ and a public Blockchain network to deploy $\mathcal{B}_{DEDU}$. At the end of the section, we provide an analysis of our proposed smart contract.
\subsubsection{Assumptions}
\begin{enumerate}
	\item We assume that there are no unintentional system failures which may affect the utilities of the CSP and the user. 
	\item We assume that the CSP, the cloud, the smart contract and the Blockchain network are available all the time.
\end{enumerate}
\subsubsection{$\mathcal{B}_{DEDU}$ contract}\label{contract:1}
$\mathcal{B}_{DEDU}$ is a contract between a CSP $c$ and a cloud user $u$. The high-level idea is that if both $c$ and $u$ are honest, then $c$ will receive the fee\footnote{From here on we call the amount paid by $u$ as fee which includes both $SF_{u}^{c}(t)$ and $EF_{u}^{c}(t)$ values} paid by $u$ and $u$ will receive the file link to access the $u$'s file in the cloud managed by $c$. The fee is computed according to rules encoded in $\mathcal{B}_{DEDU}$ contract. The clauses in the $\mathcal{B}_{DEDU}$ contract are as follows:
\begin{enumerate}[label={(\arabic*)}]
	\item\label{clause1} All parties agree on timing parameters $\tau_{p}<\tau_{c1}<\tau_{c2}$ and two payment parameters: $SF_{u}^{c}(t)$ and $EF_{u}^{c}(t)$. 
	\item\label{clause2} $c$ creates a smart contract ($\mathcal{B}_{DEDU}$) for facilitating payments for cloud storage deduplication. $c$ deploys the $\mathcal{B}_{DEDU}$ on a public Blockchain network and announces the smart contract address and smart contract ABI on a public platform (like a website/bulletin board). 
	\item\label{clause3} After verifying the contract details at the contract address, a user $u$ if willing to store data at cloud managed by $c$, has to send a request to $\mathcal{B}_{DEDU}$ along with some safety deposit $\$d$. This safe deposit is required to penalize $u$ for sending false requests. $u$'s request includes a $tag$ computed from the encrypted file, and length of the file in bits.
	\item\label{clause4} After receiving the request, $\mathcal{B}_{DEDU}$ checks whether the $tag$ sent by $u$ is received previously. If the check is valid, it will compute the fee as $(\frac{SF^{c}_{u}(t)*|d|}{n_{d}^{c}(t)} + EF_{u}^{c}(t))$ and sends this information to $u$. Otherwise it will compute the fee as $SF^{c}_{u}(t)*|d| + EF_{u}^{c}(t)$ and sends it to $u$. $|d|$ is the length of the data to be stored in bits.
	\item\label{clause5} $u$ must send the fee to $\mathcal{B}_{DEDU}$ before $\tau > \tau_{p}$. If $u$ fails then his deposit $\$d$ is sent to $c$ and the request is marked as terminated. Otherwise $u$'s deposit $\$d$ is refunded. 
	\item\label{clause6} 
	$c$ has to send the confirmation message to $\mathcal{B}_{DEDU}$ before $\tau > \tau_{c1}$ acknowledging the receipt of file or correct PoP . Otherwise, the fee paid by $u$ is refunded, and the request is marked as terminated. $c$ should send the file link with correct access rights to $u$.
	\item\label{clause7} 
	$u$ has to send the confirmation message to $\mathcal{B}_{DEDU}$ before $\tau > \tau_{c2}$ acknowledging the receipt of file link.  Otherwise the fee paid by $u$ is refunded. If $u$ has sent the confirmation message before $\tau > \tau_{c2}$ then, (1) if $u$ is the first uploader of $d$ then the fee is sent to $c$ (2) if $u$ is not the first uploader of $d$, then the $EF_{u}^{c}(t)$ part of fee is sent to $c$ and the $SF_{u}^{c}(t)$ part of fee is distributed equally among all the users who hold the file link of $d$ before $u$. In either case, the request is marked as terminated and the value of $n_{d}^{c}(t)$ is incremented.
\end{enumerate}

\begin{center}
	\footnotesize
	\begin{longtable}{|p{1.2cm}p{0.9\textwidth}|}
	\hline
	  \multicolumn{2}{|c|}{$contract$-$\mathcal{B}_{DEDU}$} \\
 \textbf{init:}  & $SF:=0$, $EF:=0$, $uTAB:=\{\}$, $k:=0$ \\	
 \textbf{Create:} & Upon receiving ("create", $SPay$, $EPay$, $interval$) from a cloud storage provider $c$ \\
 & set $SF:=SPay$, $EF:=EPay$ and $k:=interval$\\
 \textbf{Request:} & Upon receiving ("request", $tag$, $|d|$, $\$d$) from a user $u$ \\
  &assert $ledger[u] \ge \$d$ \\
  &set $ledger[u]:=ledger[u] - \$d$ \\
  & if $(tag,*,*,*) \in uTAB$ \\
  & $\:$ if  $\exists \: request(*,*,*,*,*,active,*,*) \in uTAB[tag]$ \\
  & $\:$ $\:$ set $pay:=uTAB[tag].cPay$ \\
  & $\:$ else set $pay:=SF*|d|+EF$ \\
  & else \\
  & $\:$ set $pay:=SF*|d|+EF$ \\
  & $\:$ set $numReq:=0$\\
  & $\:$ set $uTAB:=uTAB \cup (tag,numReq,*,*)$ \\
  & set $ID:=u$, $rState:=waitForPay$, $\$paid:=0$, $reqNum:=uTAB[tag].numReq$ \\
  & set  $\tau_{sub}:=\tau$, $\tau_{p}:=\tau_{sub}+k$, $\tau_{c1}:=\tau_{p}+k$, $\tau_{c2}:=\tau_{c1}+k$ \\
  & set $uTAB[tag].requests$ := $uTAB[tag].requests$ $\cup$\\ 
  &  $\ \ \ \ \ \ \ \ \ \ \ \ \ \ \ \ \ \ \ \ \ \ \ \ \ \ \ \ \ \ \ \ $ ($ID$, $\tau_{sub}$, $\tau_{p}$, $\tau_{c1}$, $\tau_{c2}$, $rState$, $pay$, $\$paid$) \\
  & set $uTAB[tag].cPay:=pay$\\
  & set $uTAB[tag].numReq:=uTAB[tag].numReq+1$ \\
  & send ("pay", $tag$, $SF$, $pay$, $reqNum$, $\tau_{p}$, $\tau_{c1}$, $\tau_{c2}$) to user $u$\\
  \textbf{Pay:} & Upon receiving ("pay",$tag$, $reqNum$, $\$pay$) from a user $u$\\
  & assert $\tau \le \tau_{p}$ \\
  & assert $uTAB[tag].requests[reqNum].ID = u$ \\
  & assert $uTAB[tag].requests[reqNum].pay \ge \$pay$ \\
  & assert $uTAB[tag].requests[reqNum].rState:=waitForPay$ \\
  & assert $ledger[u] \ge \$pay$ \\
  & set $ledger[u] = ledger[u] - \$pay$ \\
  & set $uTAB[tag].requests[reqNum].\$paid:=\$pay$ \\
  & set $ledger[u] = ledger[u] + uTAB[tag].requests[reqNum].\$d$ \\
  & set $uTAB[tag].requests[reqNum].rState:=waitForCSPConf$ \\
  \textbf{CSPConf:} & Upon receiving ("cspConf", $tag$, $reqNum$) from a cloud storage provider $c$ \\
  & assert $\tau \le \tau_{c1}$ \\
  & assert $uTAB[tag].requests[reqNum].rState=waitForCSPConf$ \\
  & set $uTAB[tag].requests[reqNum].rState:=waitForCliConf$ \\
  \textbf{usrConf:} & Upon receiving ("usrconf", $tag$, $reqNum$) from a user $u$\\
  & assert $\tau \le \tau_{c2}$ \\
  & assert $uTAB[tag].requests[reqNum].ID = u$ \\
  & assert $uTAB[tag].requests[reqNum].rState=waitForCliConf$ \\
  & $\forall$ $i \in [0,uTAB[tag].numReq-2]$\\
  & $\:$ if $uTAB[tag].requests[i].rState=active$ \\
  & $\:$ $\:$ set $activeRequests:=activeRequests+1$ \\
  & if $activeRequests=0$\\
  & $\:$ set $ledger[c]:=ledger[c]+ uTAB[tag].requests[rNum].\$paid$ \\
  & else \\
  & $\:$ set $\$rem := uTAB[tag].requests[reqNum].\$paid - EF$ \\
  & $\:$ set $\$DF:= uTAB[tag].requests[reqNum].\$paid - \$rem$ \\
  & $\:$ set $ledger[c]:=ledger[c] + \$DF$\\
  & $\:$ $\forall$ $i \in [0,activeRequests]$\\
  & $\:$ $\:$ if $uTAB[tag].requests[i].rState=active$\\
  & $\:$ $\:$ $\:$ set $ledger[uTAB[tag].requests[i].ID]:=$ $ledger[uTAB[tag].requests[i].ID]$\\ 
  & $  \ \ \ \ \ \ \ \ \ \ \ \ \ \ \ \ \ \ \ \  \ \ \ \ \ \ \ \   \ \ \ \ \ \ \ \  \ \ \ \ \ \ \ \   \ \ \ \ \ \ \ \  \ \ \ \ \ \ \ \ + \frac{\$rem}{activeRequests}$\\
  & set $uTAB[tag].cPay:= \frac{SF}{activeRequests+2}+EF$\\
  & set $uTAB[tag].requests[reqNum].rState:=active$ \\
  \textbf{Refund:} & Upon receiving ("refund", $tag$, $reqNum$) from a user $u$\\
  & assert $uTAB[tag].requests[reqNum].ID = u$ \\
  & assert $(\tau > \tau_{c1}$ $\&\&$ $uTAB[tag].requests[reqNum].rState=waitForCSPConf)$
 $||$ \\ & $\ \ \ \ \ \ \ (\tau > \tau_{c2}$ $\&\&$ $uTAB[tag].requests[reqNum].rState=waitForCliConf))$  \\
  & set $ledger[uTAB[tag].requests[reqNum].ID]$ := \\ 
  &$ \ \ \ \ \ \ \ \ \ \ \ \ \ \ \ \ \ \ ledger[uTAB[tag].requests[reqNum].ID]$ + \\
  & $\ \ \ \ \ \ \ \ \ \ \ \ \ \ \ \ \ \ uTAB[tag].requests[reqNum].\$paid$\\
  \textbf{Claim:} & Upon receiving ("claim", $tag$, $reqNum$) from a cloud storage provider $c$ \\
  & assert $\tau>\tau_{p} \: \&\& \: uTAB[tag].requests[reqNum].rState:=waitForPay$ \\
  & set $ledger[c]:=ledger[c]+uTAB[tag].requests[reqNum].\$d$ \\
  \textbf{DeLink:} & Upon receiving ("deLink", $tag$, $reqNum$) from a user $u$\\
  & assert $uTAB[tag].requests[reqNum].ID = u$ \\
  & assert $uTAB[tag].requests[reqNum].rState=active$ \\
  & set $uTAB[tag].requests[reqNum].rState:=inActive$ \\
 \hline
	\end{longtable}
\captionof{figure}{$contract$-$\mathcal{B}_{DEDU}$}
\label{fig:naiveContract}
\end{center}
\begin{figure}[!h]
	\centering
	\footnotesize
\begin{tabular}{|p{1.6cm}p{0.9\textwidth}|}
	\hline
	\multicolumn{2}{|c|}{$Proto$-$\mathcal{B}_{DEDU}$} \\
				&	Let $(KeyGen_{CE}, Encrypt_{CE}, Decrypt_{CE}, TG_{CE})$ be a secure convergent encryption scheme. \\
	\textbf{\underline{As a cloud storage provider $c$:}} & \\
	\textbf{Create:} & Upon receiving ("create", $SPay$, $EPay$, $interval$) from environment $\mathcal{E}$ \\
				& send ("create", $SPay$, $EPay$, $interval$) to $contract$-$\mathcal{B}_{DEDU}$ \\
	\textbf{CSPConf:} & Upon receiving ("file", $d$, $tag$, $reqNum$) from user $u$ \\
    				& assert that $u$ has sent $\$pay$ to $contract$-$\mathcal{B}_{DEDU}$ with the same tag \\
    				& send ("cspconf", $tag$, $reqNum$) to $contract$-$\mathcal{B}_{DEDU}$\\
    				& store $d$ and send ("link", $tag$, $reqNum$, $L$) to user $u$ \\
    				& Upon receiving ("proof", $PoP$, $tag$, $reqNum$) from user $u$ \\
    				& assert that $PoP$ is the correct proof for $tag$ \\
    				& assert that $u$ has sent $\$pay$ to $contract$-$\mathcal{B}_{DEDU}$ with the same tag \\
    				& send ("cspConf", $tag$, $reqNum$) to $contract$-$\mathcal{B}_{DEDU}$\\
    				& store $d$ or add $u$ to the user list of $d$ and send ("link", $tag$, $reqNum$, $L$) to user $u$ \\
\textbf{Claim:} & Upon receiving ("claim", $tag$, $reqNum$) from environment $\mathcal{E}$ \\
			& send ("claim", $tag$, $reqNum$) to $contract$-$\mathcal{B}_{DEDU}$ \\
\textbf{DeLink:} & Upon receiving ("deLink", $tag$, $reqNum$, $L$) from user $u$ \\
			& assert that $u$ has sent "deLink" message to $contract$-$\mathcal{B}_{DEDU}$ with same $tag$ and $reqNum$ \\
			& disable the link $L$ for $u$ or remove $d$ if no other users have active links to $d$ \\
\textbf{\underline{As a user $u$:}} & \\    				
	\textbf{Request:} & Upon receiving ("request", $d$) from environment $\mathcal{E}$ \\
				  & compute $K \leftarrow KeyGen_{CE}(d)$\\
				  & $C\leftarrow Encrypt_{CE}(K,d)$ \\ 
				  & $tag\leftarrow TG_{CE}(C)$ \\
				  & send ("request", $tag$, $|d|$, $\$d$) to $contract$-$\mathcal{B}_{DEDU}$ \\
	\textbf{Pay:} & Upon receiving ("pay", $tag$, $SF$, $pay$, $reqNum$, $\tau_{p}$, $\tau_{c1}$, $\tau_{c2}$) from $contract$-$\mathcal{B}_{DEDU}$\\
			& assert that a request message has sent earlier with the same $tag$. \\
			&  send $("pay", tag, reqNum, \$pay)$ to $contract$-$\mathcal{B}_{DEDU}$ \\
			& if $SF*|d| = pay$ send ("file", $d$, $tag$, $reqNum$) to $c$\\
			& else send ("proof", $PoP$, $tag$, $reqNum$) to $c$ \\
	\textbf{usrConf:} & Upon receiving ("link", $tag$, $reqNum$ $L$) from $c$ \\
				& assert that $L$ is a correct link to File $d$ \\
				& assert that $c$ has sent "cspConf" message to $contract$-$\mathcal{B}_{DEDU}$ with the same $tag$ \\
				& send ("usrConf", $tag$, $reqNum$) to $contract$-$\mathcal{B}_{DEDU}$ \\
\textbf{Refund:} & Upon receiving ("refund", $tag$, $reqNum$) from environment $\mathcal{E}$ \\
			& send ("refund", $tag$, $reqNum$) to $contract$-$\mathcal{B}_{DEDU}$ \\ 		
\textbf{DeLink:} & Upon receiving ("deLink", $tag$, $reqNum$) from environment $\mathcal{E}$ \\
			& send ("deLink", $tag$, $reqNum$) to $contract$-$\mathcal{B}_{DEDU}$ \\
			& send ("deLink", $tag$, $reqNum$, $L$) to $c$ \\
	\hline
	\end{tabular}
\caption{Client-side programs for $contract$-$\mathcal{B}_{DEDU}$}
\label{fig:naiveProto}
\end{figure}
\subsubsection{Analysis of $\mathcal{B}_{DEDU}$}
In $\mathcal{B}_{DEDU}$, the parties agree on different timing parameters $\tau_{p}<\tau_{c1}<\tau_{c2}$. These timing parameters are required to enforce timely computation and also to avoid locking of the funds if one of the party refuses to move forward in the protocol. The contract $\mathcal{B}_{DEDU}$ can always query the underlying Blockchain for current time\footnote{most smart contracts use block number or block timestamp as a timer.} which is different from the real-world timer. 
Clause \ref{clause4} ensures the user that the fee is calculated according to the deduplication rate. Clause \ref{clause5} ensures the user that if he sends the fee before $\tau_{p}$, then his deposit is refunded; otherwise, his deposit is forfeited. According to clause \ref{clause6}, in order to get the fee, the CSP has to send an acknowledgement before $\tau_{c1}$. Otherwise, it is deemed as either the CSP do not want to store the user's data, or the user has not sent the data or correct proof-of-ownership. According to clause \ref{clause7} in order to store his data at cloud and get access rights, the user has to send the acknowledgement of the receipt of file link before $\tau_{c2}$. Otherwise, it is deemed as either the user has not received the file link or the user is acting maliciously after receiving the file link. In this case, the CSP does not store the file or do not provide the access rights to the user. Clause \ref{clause7} also guarantees that if $u$ is not the first uploader of a data $d$, then the $SF_{u}^{c}(t)$ part of the fee paid by $u$ is equally distributed among all the users who hold the file link of $d$. This distribution of storage fee among users is unique to our scheme, and it guarantees that all users pay the same storage fee irrespective of when they submit their request. A formal smart contract encoding the clauses discussed above is presented in Figure \ref{fig:naiveContract} and a protocol to interact with the smart contract is presented in Figure \ref{fig:naiveProto}.

A cloud service provider $c$ initializes the parameters $SF$ and $EF$ by invoking Create functionality. $c$ sets these parameters according to current storage costs and utility. If the storage costs vary in future, he can change $SF$ and $EF$ values according to new storage costs. The parameter $interval$ is required  to compute timing parameters $\tau_{p}$, $\tau_{c1}$ and $\tau_{c2}$. These timing parameters are required for timely computation of protocol and avoiding indefinite locking of funds in the contract. A user $u$ sends his storage request to $\mathcal{BC}$ invoking Request functionality. His request consists of parameters like $tag$, $|d|$ and $\$d$. $tag$ is computed using a convergent encryption algorithm, and $|d|$ is the length of the file requested for storage and $\$d$ is a safety deposit. The Request functionality computes the pay in two ways. If the $tag$ sent by $u$ is already exists and is in active state at cloud maintained by $c$, then $pay$ is computed according to the current discounted storage fee of file $F$ with tag $tag$. Otherwise, the pay is computed as $pay = SF * |d| + EF$. Depending on the behavior of the user, there are two cases as follows:
\begin{myEnumerate}[label={}, leftmargin = 0pt]
	\item \textbf{Case 1}: $u$ has sent the storage fee to $\mathcal{BC}$ invoking Pay functionality. Depending on the behavior of the user, there are five cases as follows:
	\begin{myEnumerate}[label={}]
		\item \textbf{Case 1.1}: $u$ has sent the correct file $d$ \footnote{The correctness of the file is verified using the $tag$, based on which the payment is computed during execution of Request functionality.} to $c$. Depending on the behavior of $c$, there are two cases as follows:
		\begin{myEnumerate}[label={}]
			\item \textbf{Case 1.1.1}: $c$ has sent the confirmation message to $\mathcal{BC}$ invoking CSPConf functionality. Depending on the behavior of $c$, there are two cases as follows:
			\begin{myEnumerate}[label={}]
				\item \textbf{Case 1.1.1.1}: $c$ has sent the file link to $u$. Depending on the behavior of the $u$, there are two cases as follows:
				\begin{myEnumerate}[label={}]
					\item \textbf{Case 1.1.1.1.1}: $u$ has sent the confirmation message to $\mathcal{BC}$ invoking usrConf functionality. In this case, if file $d$ is not previously stored at $c$, then all the $pay$ is sent to $c$. Otherwise, the number of owners currently having a link to file $d$ is calculated, and $pay$ is divided into $rem$ and $DF$ components. $rem$ is distributed among the file owners equally, and $DF$ is sent to $c$. The new $pay$ to be paid by the next deduplication requester is computed and stored at contract storage.
					\item \textbf{Case 1.1.1.1.2}: $u$ has failed to send the confirmation message. This case occurs when $u$ has not received the file link from $c$ or $u$ intentionally / unintentionally fail to send a confirmation message. In this case, the $c$ can invalidate the link sent to $u$ and $u$ can send a transaction to $\mathcal{BC}$ invoking Refund functionality, which refunds $pay$ to $u$.
				\end{myEnumerate}
				\item \textbf{Case 1.1.1.2}:  $c$ has failed to send the file link to $u$. This case is similar to case 1.1.1.1.2 where $u$ can obtain a refund by invoking Refund functionality.
			\end{myEnumerate}
			\item \textbf{Case 1.1.2}: $c$ has failed to send the confirmation message to $\mathcal{BC}$. This case is similar to case 1.1.1.1.2 where $u$ can obtain a refund by invoking Refund functionality.
		\end{myEnumerate}
		\item \textbf{Case 1.2}: $u$ has sent the incorrect file to $c$. In this case, the $c$ discards the $u$'s request and will not send any further transactions. $u$ can obtain its $pay$ invoking Refund functionality.
		\item \textbf{Case 1.3}: $u$ has sent the correct proof-of-possession to $c$. In this case, $c$ adds $u$ to the list of owners of file $d$ and sends the file link to $u$. From now on this case proceeds similarly to case 1.1.1.1.
		\item \textbf{Case 1.4}: $u$ has sent the incorrect proof-of-possession to $c$. This case is similar to case 1.2.
		\item \textbf{Case 1.5}: $u$ neither sends file nor proof-of-possession to $c$.  This case is similar to case 1.2.
	\end{myEnumerate}
	\item \textbf{Case 2}: $u$ has failed to send the storage fee. In this case, $c$ sends a transaction to $\mathcal{BC}$ invoking Claim functionality to claim $u$'s deposit.
\end{myEnumerate}
\subsubsection{Proofs of $\mathcal{B}_{DEDU}$}
\begin{theorem}
Proposed scheme satisfies correctness
\end{theorem}
\begin{proof}
	Let $contract$-$\mathcal{B}_{DEDU}$ is deployed on a Nakamoto-style Blockchain network using Proof-of-Work as a consensus algorithm. Let $b$ be the current block of the blockchain which is extended by blocks $b_{1}$ and $b_{2}$. Let $tx_{request}$ and $tx_{pay}$ be the request and pay transactions initiated by a user $u$ which are eventually embedded in $b_{1}$ and $b_{2}$ respectively. We will prove the correctness by considering the following cases. \par
	Case 1: A user $u$ influences the execution of the request functionality to decrease the fee $f$ to be paid by him. This can happen if $u$ assumes the role of a miner and pre-mines two blocks $b_{1}^{'}$ and $b_{2}^{'}$ that contains the $tx_{request}^{'}$ and $tx_{pay}^{'}$ transactions respectively with a modified fee $f^{'}$. Remember in this case $u$ does not broadcast $tx_{request}^{'}$ to the network. Now $u$ broadcast $b_{1}^{'}$ extending $b$ and $b_{2}^{'}$ extending $b_{1}^{'}$ to the Blockchain network. All the other miners in the Blockchain network verifies every transaction in $b_{1}^{'}$ and extends $b_{1}^{'}$ if and only if all the outputs of every transaction in $b_{1}^{'}$ are correct. As the output of $tx_{request}$ is $f$ but not $f^{'}$ the miners discard the block $b_{1}^{'}$. As $b_{2}^{'}$ is extending the wrong block $b_{1}^{'}$, it is also discarded. Miners with at least 51\% of hash-rate cumulatively required to guarantee the correctness of the transactions in the block. \par
	Case 2: $u$ broadcasts $tx_{request}$ but pre-mines $b_{2}^{'}$ that contains $tx_{pay}^{'}$ transaction with modified fee $f^{'}$. In this case, as the actual fee $f$ to be paid by $u$ is already stored in the contract storage, while verifying the $tx_{pay}^{'}$ in $b_{2}^{'}$, the $f^{'}$ received through $tx_{pay}^{'}$ is compared against the stored value. The comparison will fail and the block $b_{2}^{'}$ is discarded. \par
	Case 3: A cloud storage provider $c$ influences the execution of request functionality to increase the fee to be paid by $u$. This case is similar to Case 1 except that now $c$ assumes the role of a miner and broadcasts $b_{1}^{'}$ consisting the modified fee. \par
	Similarly, all the transactions with $\mathcal{B}_{DEDU}$ are executed correctly; otherwise, those transactions are rejected. Nevertheless,  $c$ or $u$ can influence some of the miners to include the wrong blocks to their local ledger and generate new blocks extending these wrong blocks. However, $c$ or $u$ should accumulate more than 50\% of hash rate to make the entire network to accept wrong blocks which is a difficult task unless they have large mining pools under their control \cite{kroll2013economics}. 
	
	In summary, as the miners in public Blockchain networks are reasonably assumed to act honestly for the common good and follow the rules encoded in consensus algorithm, it is difficult for $c$ or $u$ to make the network to accept wrong blocks. Considering the above cases, our scheme satisfies correctness as defined in Section \ref{def:pb}.
\end{proof}
\begin{theorem}\label{th:fairpayments}
Proposed scheme satisfies fair payments
\end{theorem}
\begin{proof}
	We prove fairness by considering the following cases. \par
	Case 1: $u$ is malicious and aborts after learning the fee he needs to pay. In this case, according to $\mathcal{B}_{DEDU}$ contract clause \ref{clause5}, $u$ forfeits his deposit, and his data is not stored in the cloud. Here, the data of the $u$ cannot be stored in the cloud unless he pays the fee. Thus the fairness holds\par    
	Case 2: $u$ fails to send the data $d$ or sends incorrect $PoP$ to $c$. In this case, according to $\mathcal{B}_{DEDU}$ contract clause \ref{clause6}, $c$ refuses to acknowledge the receipt of $d$/$PoP$ and the fee paid by $u$ is refunded. Here, the $u$ cannot obtain the storage link if he fails to send the $d$ or the correct $PoP$. Thus the fairness holds.\par
	Case 3: $c$ is malicious and do not acknowledge the receipt of $d$ or correct $PoP$. This case similar to Case 2. The fee paid by $u$ is refunded. Here $c$ cannot obtain the fee without acknowledging the receipt of $d$ or correct $PoP$. Thus the fairness holds.\par
	Case 4: $c$ is malicious and do not send the file link to $u$. In this case, according to $\mathcal{B}_{DEDU}$ contract clause \ref{clause7} $u$ will not acknowledge the receipt of file link. Then the fee paid by $u$ is refunded. Here, $c$ cannot obtain the fee without sending the file link to $u$. Thus the fairness holds. \par
	Case 5: $u$ is malicious and do not acknowledge the receipt of the file link. This case is similar to case 4. The fee paid by $u$ is refunded. If $c$ do not receive the fee, it disables the file link sent to $u$. Here, $u$ cannot store his data without acknowledging the receipt of file link. Thus the fairness holds.\par
	In summary, considering the above cases, our scheme holds fairness.
\end{proof}
\begin{theorem}
	Proposed scheme satisfies uniform payments
\end{theorem}
\begin{proof}
According to clause \ref{clause4}, the first uploader pays a fee of $SF_{u}^{c}(t)+EF_{u}^{c}(t)$. The second uploader pays a fee of $\frac{SF_{u}^{c}(t)}{2}+EF_{u}^{c}(t)$. Due to the second uploader, the first uploader gets a refund of $\frac{SF_{u}^{c}(t)}{2}$. So at the end each uploader pays a fee of $\frac{SF_{u}^{c}(t)}{2}+EF_{u}^{c}(t)$. 

The $n^{th}$ uploader pays a fee of $\frac{SF_{u}^{c}(t)}{n}+EF_{u}^{c}(t)$ and due to the $n^{th}$ uploader each of the previous $n-1$ uploaders receive a refund of $\frac{SF_{u}^{c}(t)}{n*(n-1)}$. At the end of $n^{th}$ user's payment, the first user pays $SF_{u}^{c}(t)+EF_{u}^{c}(t)-\frac{SF_{u}^{c}(t)}{2}-\frac{SF_{u}^{c}(t)}{6},...,-\frac{SF_{u}^{c}(t)}{n*(n-1)}=\frac{SF_{u}^{c}(t)}{n}+EF_{u}^{c}(t)$. The second user pays $\frac{SF_{u}^{c}(t)}{2}+EF_{u}^{c}(t)-\frac{SF_{u}^{c}(t)}{6},...,-\frac{SF_{u}^{c}(t)}{n*(n-1)}=\frac{SF_{u}^{c}(t)}{n}+EF_{u}^{c}(t)$. Similarly the $(n-1)^{th}$ user pays $\frac{SF_{u}^{c}(t)}{n-1}+EF_{u}^{c}(t)-\frac{SF_{u}^{c}(t)}{n*(n-1)}=\frac{SF_{u}^{c}(t)}{n}+EF_{u}^{c}(t)$.
So, every user having a data $d$ opting for dedup with $\mathcal{B}_{DEDU}$ at $c$ pays the same fee. 

In summary, the proposed scheme satisfies uniform payments irrespective of when a user submits his/her request. 
\end{proof}
%

\section{Proposed Inter-CSP deduplication scheme}
In the previous section, we proposed a dedup scheme with a single CSP whereas in this section, we propose a Blockchain-based inter-CSP deduplication scheme, which consists a root level smart contract $\mathcal{B}_{I\text{-}DEDU}$, a smart contract $\mathcal{B}_{DEDU}$ for each CSP, protocols to interact with $\mathcal{B}_{DEDU}$, and a public Blockchain network to deploy the smart contracts.
\subsection{Assumptions}
\begin{enumerate}
	\item We assume all the CSPs charges the same $SF_{u}^{c}(t)$ and $EF_{u}^{c}(t)$ values. 
	\item We assume that no CSP will collect extra fee other than $SF_{u}^{c}(t)$ or $EF_{u}^{c}(t)$ for inter CSP deduplication. Inter-CSP deduplication is same as single CSP deduplication from the cloud users perspective.
	\item A CSP $c_{1}$ pays a fee of $AF^{c}(t)$ to a CSP $c_{2}$ for accessing the data stored at $c_{2}$.
	\end{enumerate}
\subsection{$\mathcal{B}_{I\text{-DEDU}}$}
The inter-CSP deduplication is similar to single CSP deduplication except that now a root-level smart contract has the information about tags of data stored in different CSPs. The clauses in $\mathcal{B}_{I\text{-}DEDU}$ are as follows:
\begin{enumerate}[label={(\arabic*)}]
		\item An organization $\mathcal{O}$ (like a consortium of CSPs) creates a smart contract to facilitate inter CSP deduplication scheme. $\mathcal{O}$ designs and deploys the $\mathcal{B}_{I\text{-}DEDU}$ on a public Blockchain network and shares the smart contract address and ABI with all the registered CSPs\footnote{We assume a registration phase is executed before deploying the smart contract and only registered CSPs can exchange messages with $\mathcal{B}_{I\text{-}DEDU}$}.
		\item Each CSP $c_{i}$ also deploys a smart contract $\mathcal{B}_{DEDU}^{i}$ on a public Blockchain network and announces the smart contract address and smart contract ABI on a public platform (like a website/bulletin board) and also registers $\mathcal{B}_{DEDU}^{i}$ address with $\mathcal{B}_{I\text{-}DEDU}$.
		\item A user $u$ sends request to $\mathcal{B}_{DEDU}^{i}$ similar to clause (3) in Section \ref{contract:1}.
		\item $\mathcal{B}_{DEDU}^{i}$ performs checks similar to clause (4) in Section \ref{contract:1}. If the check is valid, it will compute the fee to be paid by $u$ according to the dedup rate and sends this information to $u$. Otherwise it will send a request containing $tag$ to $\mathcal{B}_{I\text{-}DEDU}$.
		\item $\mathcal{B}_{I\text{-}DEDU}$ checks whether any CSP is holding data with the same $tag$. If check is not valid, then it will return $N/A$ to $\mathcal{B}_{DEDU}^{i}$. Then $\mathcal{B}_{DEDU}$ will compute the fee as $SF_{u}^{c}(t)*|d|+ EF_{u}^{c}(t)$ and send it to $u$. Otherwise $\mathcal{B}_{I\text{-}DEDU}$ will send the information about the smart contract $\mathcal{B}_{DEDU}^{j}$ and CSP $c_{j}$ which is holding the data with $tag$.   
		\item If $\mathcal{B}_{DEDU}^{i}$ receives the info about $c_{j}$, then a request message is sent to $\mathcal{B}_{DEDU}^{j}$ from $\mathcal{B}_{DEDU}^{i}$. Then $\mathcal{B}_{DEDU}^{j}$ computes the fee according to the dedup rate and sends it to $\mathcal{B}_{DEDU}^{i}$, which will be then forward to user $u$ by $\mathcal{B}_{DEDU}^{i}$.
		\\ From here on, we assume that inter CSP deduplication is found and the contract proceeds as follows:
		\item If $\tau > \tau_{p}$ and $u$ has not sent the fee to $\mathcal{B}_{DEDU}^{i}$, then $\mathcal{B}_{DEDU}^{i}$ also won't send fee to $\mathcal{B}_{DEDU}^{j}$. $u$'s deposit is forfeited. This deposit is sent to $c_{j}$ and the request is marked as terminated by both $\mathcal{B}_{DEDU}^{j}$ and $\mathcal{B}_{DEDU}^{i}$.
		\item If $\tau > \tau_{c1}$ and $c_{j}$ has not sent the confirmation message to $\mathcal{B}_{DEDU}^{j}$, then the fee paid by $u$ is refunded and the request is marked as terminated by both $\mathcal{B}_{DEDU}^{j}$ and $\mathcal{B}_{DEDU}^{i}$.
		\item If $\tau > \tau_{c2}$ and $u$ has not sent the confirmation message to $\mathcal{B}_{DEDU}^{i}$, then $\mathcal{B}_{DEDU}^{i}$ won't send confirmation to $\mathcal{B}_{DEDU}^{j}$. The fee paid by $u$ is refunded. Else, a fee $AF^{c}(t)$ is sent to $c_{j}$ from $c_{i}$ and the request is marked as terminated by both $\mathcal{B}_{DEDU}^{j}$ and $\mathcal{B}_{DEDU}^{i}$. 
\end{enumerate}
The analysis of $\mathcal{B}_{I\text{-}DEDU}$ is similar to $\mathcal{B}_{DEDU}$ and all the properties satisfied by $\mathcal{B}_{DEDU}$ are also satisfied by $\mathcal{B}_{I\text{-}DEDU}$. The formal smart contract $contract$-$\mathcal{B}_{I\text{-}DEDU}$ is presented in Figure \ref{fig:contractInter}.
\begin{figure}
	\centering
	\begin{tabular}{|p{1.5cm}p{0.8\textwidth}|}
		\hline
		\multicolumn{2}{|c|}{$contract$-$\mathcal{B}_{I\text{-}DEDU}$} \\
		\textbf{Init:} & $list:=\{\}$, $tags:=\{\}$\\
		\textbf{Register:} & Upon receiving ("register", $\mathcal{B}_{DEDU}$, $info$) from a cloud storage provider $c$ \\
					 & assert $(c,*,*) \notin list$ \\
					 & set $list := list \cup (c,\mathcal{B}_{DEDU},info)$ \\
		\textbf{setTag:} & Upon receiving ("setTag", $tag$) from a contract $\mathcal{B}_{DEDU}$ \\
		& assert $tag \notin tags $ \\
		& set $CSP:= list[\mathcal{B}_{DEDU}].c$ \\
		& set $info:=list[\mathcal{B}_{DEDU}].info$ \\
		& set $tags:=tags \cup (tag, \mathcal{B}_{DEDU}, CSP, info)$ \\
		\textbf{getTag:} & Upon receiving ("getTag", $tag$) from a contract $\mathcal{B}_{DEDU}$ \\
					 & assert $(*,\mathcal{B}_{DEDU},*) \in list$ \\
					 & if $(tag,*,*,*) \in tags $ \\
					 &  $\:$ send ("tagFound", $tag$, $tags[tag].\mathcal{B}_{DEDU}$, $tags[tag].CSP$, $tags[tag].info$) to $\mathcal{B}_{DEDU}$\\
					 & else send ("N/A", $tag$) to $\mathcal{B}_{DEDU}$ \\
		\hline
	\end{tabular}
\caption{$contract$-$\mathcal{B}_{I\text{-}DEDU}$}
\label{fig:contractInter}
\end{figure}

\section{Implementation}
We have written the contracts in solidity 0.8.0 using truffle framework. We have used a 2.50 GHz intel core i5 processor and a 4 GB RAM machine to run our experiments. We have deployed the proposed contract in a private Ethereum network (https://www.trufflesuite.com/ganache) which mimic public Ethereum network. However, our goal of this implementation is to deploy the contract in public Blockchain networks in real scenarios.
\subsection{Implementation of $\mathcal{B}_{DEDU}$}
We have tested the proposed smart contract $\mathcal{B}_{DEDU}$ multiple times, and each transaction's transactional cost and its equivalent financial cost is shown in Table \ref{tab:costs}. Observe that the contract deployment transaction consumes a large amount of gas; however, this is a one time cost for CSP.  Next, the create and request functionalities also consumes a large amount of gas due to the modification of contract storage variables. Storing data in a contract is an expensive operation in Ethereum. As the usrConf function executes main tasks like computing dedup rate, sending the storage fee to CSP, and sending refunds to users, the transaction to usrConf also consumes a large amount of gas. The gas consumption of usrConf varies and depends on the number of users opted for deduplication before the transaction initiator's call. We listed the gas consumption of usrConf function in Figure \ref{fig:graphusrConf}. Observe that the gas consumption increases with the increase in dedup rate, and it reaches more than the block gas limit.
\begin{table}
	\centering
	\begin{tabular}{ccc}
		\hline
		\textbf{Function} & \textbf{Caller} & \textbf{Gas cost}   \\ \hline
		Init (contract deployment) & CSP & 187467  \\
		Create & CSP & 143464 \\
		Request & User & 161168 \\
		Pay & User & 66558  \\
		CSPConf & CSP & 31457 \\
		Refund & User & 31779 \\
		Claim & CSP & 31549\\
		DeLink & User & 29318 \\
		\hline
	\end{tabular}
	\caption{Costs of interacting with $\mathcal{B}_{DEDU}$ contract.}
\label{tab:costs}
\end{table}
\begin{figure}[!h]
	\centering
	\begin{tikzpicture}
	\begin{axis}[
	xlabel = Number of users,
	ylabel = Gas consumption
	]
	\addplot coordinates {(1,191146)(50,677626)(100,1285726)(150,1893826)(200,2501926)(250,3110026)(300,3718126)(350,4326226)(400,4934326)(450,5542426)(500,6150526)};
	\end{axis}
	\end{tikzpicture}
	\caption{Costs of interacting with usrConf functionality}
	\label{fig:graphusrConf}
\end{figure}
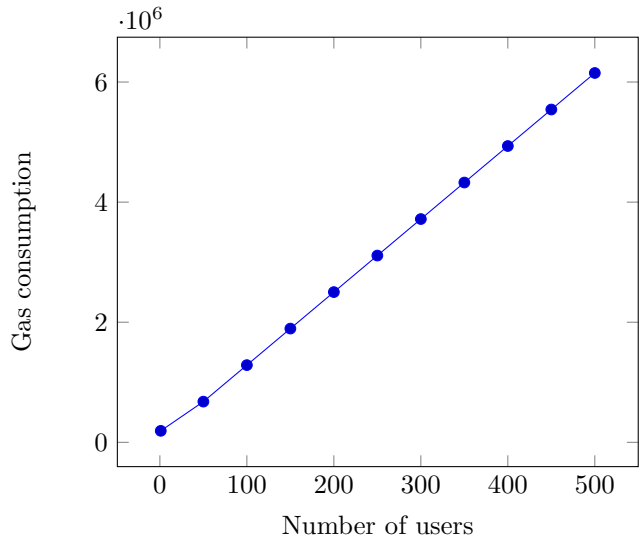

\subsection{Experiment 1: Finding utility of the users and the CSP by varying $n_{d}^{c}(t)$ and $EF_{u}^{c}(t)$}
We have conducted experiments with the values shown in Table \ref{tab:parametersettings} adopted from \cite{gao2016game,liang2019game}. The utility of the users and the CSP are shown in Figure \ref{fig:graphsUserUtility} and Figure \ref{fig:graphsCSPUtility} respectively. We varied $EF_{u}^{c}(t)$ from 10\% to 50\% of $SF_{u}^{c}(t)$ in both the figures. We also varied the deduplication rate $n_{d}^{c}(t)$ as 10\%, 50\%, 90\%, 100\% of $N_{d}^{c}(t)$ . Observe that both the users and CSP obtains non-negative utilities when opted for dedup with $\mathcal{B}_{DEDU}$. This shows that our proposed scheme is individually rational.

 In Figure \ref{fig:graphsUserUtility}, the utilities of the users decreases as the $EF_{u}^{c}(t)$ increases. 
 However, for any constant $EF_{u}^{c}(t)$ value the average utility of users increases with increase in dedup rate. This result agrees with the general notion of increase in dedup rate increases the average utility of users. Figure \ref{fig:graphsUserUtility} also shows that the user is incentive-compatible that is $U_{u}^{1}(t)>U_{u}^{0}(t)$ 
 $\forall n_{d}^{c}(t) >1$. The results show that the proposed scheme is not incentive compatible for $n_{d}^{c}(t)=1$, due to the new parameter $EF_{u}^{c}(t)$. To make the user incentive compatible when $n_{d}^{c}(t)=1$, the $EF_{u}^{c}(t)$ value should be set to zero.
 
 Figure \ref{fig:graphsCSPUtility} shows that the CSP is not incentive compatible until $EF_{u}^{c}(t) \ge 36\%$ of $SF_{u}^{c}(t)$. This result is in line with \eqref{eq:minCF} when the values in Table \ref{tab:parametersettings} are substituted.
\begin{table}[!h]
	\centering
	\begin{tabular}{c|c}
		\hline
		\textbf{Parameter} & \textbf{Value in Ether}  \\ \hline \hline
		$P_{u}(t)$    &    2.165     \\
		$SF_{u}^{c}(t)$ & 0.165  \\
		$SC_{c}^{u}(t)$&  0.1  \\
		$AF_{c}(t)$ & 0.1\\
		\hline
	\end{tabular}
	\caption{Experiment Settings.}
	\label{tab:parametersettings}
\end{table}
\pgfplotsset{width=6cm,compat=1.3}
\begin{figure}[!h]
	\begin{center}
		\begin{tikzpicture}
		\begin{axis}[
		legend entries = {$U_{u}^{0}(t)$, $U_{u}^{1}(t)(n_{d}^{c}(t)=10\%)$, $U_{u}^{1}(t)(n_{d}^{c}(t)=50\%)$, $U_{u}^{1}(t)(n_{d}^{c}(t)=90\%)$, $U_{u}^{1}(t)(n_{d}^{c}(t)=100\%)$ },
		legend to name = named,
		legend columns = 3,
		xlabel = Number of users,
		ylabel = \footnotesize{Average utility of users},
		title = {$EF_{u}^{c}(t)= 10\% \text{ of } SF_{u}^{c}(t)$},
		]		
		\addplot coordinates {(10,2)(20,2)(30,2)(40,2)(50,2) (60,2)(70,2)(80,2)(90,2)(100,2)};
		\addplot coordinates {(10,1.999)(20,2.007)(30,2.009666667)(40,2.011)(50,2.0118)
			(60,2.012166667)(70,2.012571429)(80,2.012875)(90,2.013111111)(100,2.0133)};
		\addplot coordinates {(10,2.059)(20,2.0665)(30,2.069333333)(40,2.0705) (50,2.0714) (60,2.071833333) (70,2.072285714)(80,2.0725)(90,2.072777778)(100,2.0729)};
		\addplot coordinates {(10,2.118)(20,2.126)(30,2.128666667)(40,2.13)(50,2.131) (60,2.1315) (70,2.131857143)(80,2.132125)(90,2.132333333)(100,2.1325)	};
		\addplot coordinates {(10,2.133)(20,2.141)(30,2.143666667)(40,2.145)(50,2.1458) (60,2.146333333) (70,2.146714286)(80,2.147)(90,2.147222222)(100,2.1474)	};
		\end{axis}
		\end{tikzpicture}
		\begin{tikzpicture}
		\begin{axis}[
		xlabel = Number of users,
		ylabel = \footnotesize{Average utility of users},
		title = {$EF_{u}^{c}(t)= 20\% \text{ of } SF_{u}^{c}(t)$},
		]		
		\addplot coordinates {(10,2)(20,2)(30,2)(40,2)(50,2) (60,2)(70,2)(80,2)(90,2)(100,2)};
		\addplot coordinates {(10,1.997)(20,2.005)(30,2.008)(40,2.00925)(50,2.01)(60,2.0105) (70,2.01085)(80,2.01125)(90,2.01144)(100,2.0116)	};
		\addplot coordinates {(10,2.05)(20,	2.058)(30,2.060666667)(40,2.062)(50,2.0628) (60,2.063333333)(70,2.063714286)(80,2.064)(90,2.064222222)(100,2.0644)};
		\addplot coordinates {(10,2.103)(20,2.111)(30,2.113333333)(40,2.11475)(50,2.1156) (60,2.116166667)(70,2.116571429)(80,2.11675)(90,2.117)(100,2.1172)};
		\addplot coordinates {(10,2.116)(20,2.124)(30,2.126666667)(40,2.128)(50,2.1288) (60,2.129333333) (70,2.129714286)(80,2.13)(90,2.130222222)(100,2.1304)};
		\end{axis}
		\end{tikzpicture}
		\begin{tikzpicture}
		\begin{axis}[
		xlabel = Number of users,
		ylabel = \footnotesize{Average utility of users},
		title = {$EF_{u}^{c}(t)= 30\% \text{ of } SF_{u}^{c}(t)$},
		]		
		\addplot coordinates {(10,2)(20,2)(30,2)(40,2)(50,2) (60,2)(70,2)(80,2)(90,2)(100,2)};
		\addplot coordinates {(10,1.996)(20,2.0035)(30,2.006333333)(40,2.0075)(50,2.0084) (60,2.009)(70,2.009285714)(80,2.009625)(90,2.009777778)(100,2.01)};
		\addplot coordinates {(10,2.042)(20,2.05)(30,2.052666667)(40,2.054)(50,2.0548) (60,2.055333333)(70,2.055714286)(80,2.056)(90,2.056222222)(100,2.0564)};
		\addplot coordinates {(10,2.088)(20,2.0965)(30,2.099)(40,2.1005)(50,2.1012) (60,2.101666667)(70,2.102142857)(80,2.102375)(90,2.102666667)(100,2.1028)};
		\addplot coordinates {(10,2.1)(20,2.108)(30,2.110666667)(40,2.112)(50,2.1128) (60,2.113333333)(70,2.113714286)(80,2.114)(90,2.114222222)(100,2.1144)};
		\end{axis}
		\end{tikzpicture}
		\begin{tikzpicture}
		\begin{axis}[
		xlabel = Number of users,
		ylabel = \footnotesize{Average utility of users},
		title = {$EF_{u}^{c}(t)= 40\% \text{ of } SF_{u}^{c}(t)$},
		]	
		\addplot coordinates {(10,2)(20,2)(30,2)(40,2)(50,2) (60,2)(70,2)(80,2)(90,2)(100,2)};
		\addplot coordinates {(10,1.994)(20,2.002)(30,2.004666667)(40,2.006)(50,2.0068) (60,2.007166667)(70,2.007571429)(80,2.007875)(90,2.008111111)(100,2.0083)};
		\addplot coordinates {(10,2.034)(20,2.0415)(30,2.044333333)(40,2.0455)(50,2.0464)  (60,2.046833333)(70,2.047285714)(80,2.0475)(90,	2.047777778)(100,2.0479)};
		\addplot coordinates {(10,2.073)(20,2.081)(30,2.083666667)(40,2.085)(50,2.086) (60,2.0865)(70,2.086857143) (80,2.087125)(90,2.087333333)(100,2.0875)};
		\addplot coordinates {(10,2.083)(20,2.091)(30,2.093666667)(40,2.095)(50,2.0958) (60,2.096333333)(70,2.096714286)(80,2.097)(90,2.097222222)(100,2.0974)};	
		\end{axis}
		\end{tikzpicture}
		\begin{tikzpicture}
		\begin{axis}[
		xlabel = Number of users,
		ylabel = \footnotesize{Average utility of users},
		title = {$EF_{u}^{c}(t)= 50\% \text{ of } SF_{u}^{c}(t)$},
		]	
		\addplot coordinates {(10,2)(20,2)(30,2)(40,2)(50,2) (60,2)(70,2)(80,2)(90,2)(100,2)};
		\addplot coordinates {(10,1.992)(20,2.0005)(30,2.003)(40,2.00425)(50,2.0052) (60,2.005666667)(70,2.006)(80,2.00625)(90,2.006555556)(100,2.0067)};
		\addplot coordinates {(10,2.026)(20,2.0335)(30,2.036333333)(40,2.0375)(50,2.0384) (60,2.038833333)(70,2.039285714)(80,2.0395)(90,2.039777778)(100,2.0399)	};
		\addplot coordinates {(10,2.059)(20,2.0665)(30,2.069333333)(40,2.07075)(50,2.0716) (60,2.072)(70,2.072428571)(80,2.07275)(90,2.072888889)(100,2.0731)};
		\addplot coordinates {(10,2.067)(20,2.075)(30,2.077666667)(40,2.079)(50,2.0798) (60,2.080333333)(70,2.080714286)(80,2.081)(90,2.081222222)(100,2.0814)};	
		\end{axis}
		\end{tikzpicture}
		\\ \ref{named}
	\end{center}
	\caption{The effect of $EF_{u}^{c}(t)$ and $n_{d}^{c}(t)$ on average utility of the cloud users}
	\label{fig:graphsUserUtility}
\end{figure}
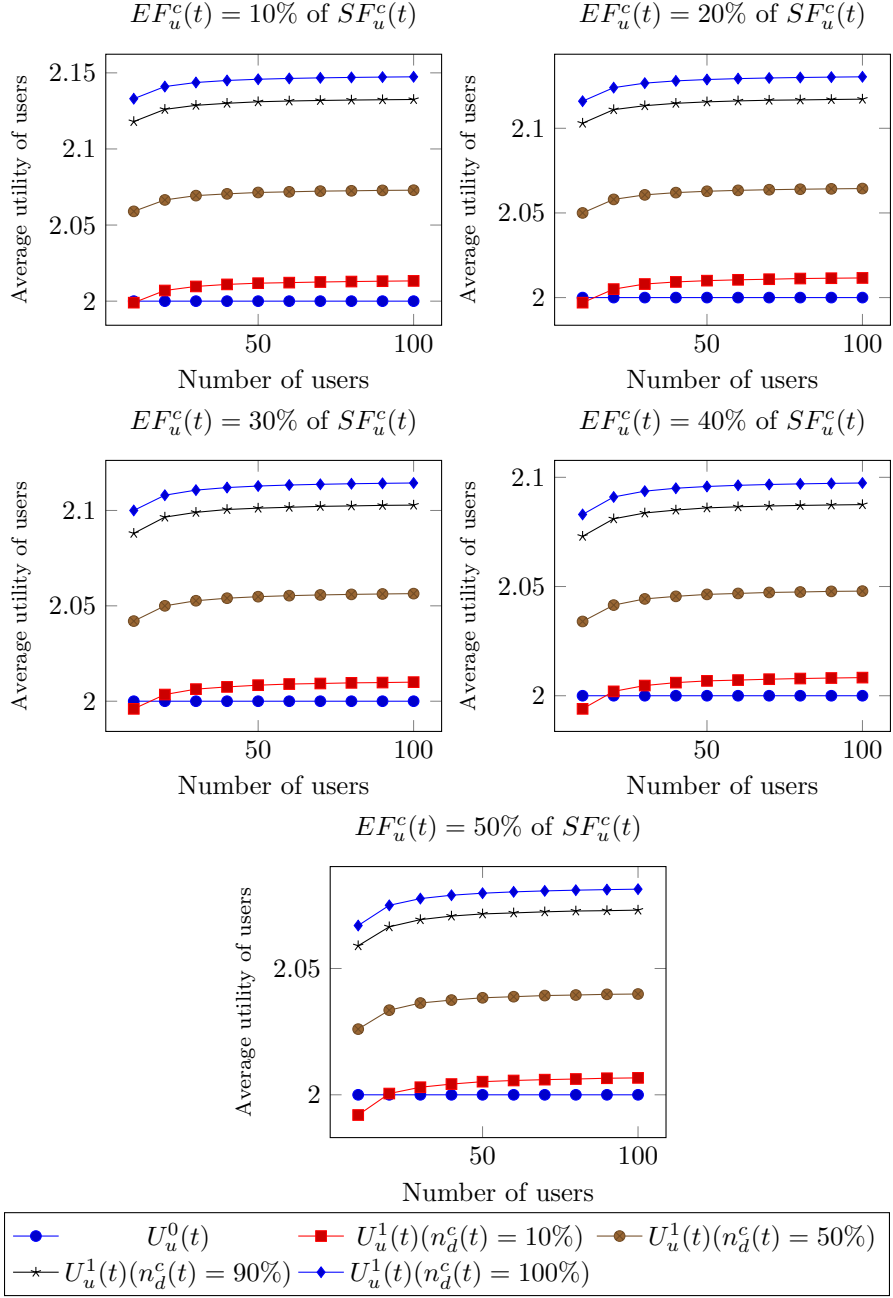

\begin{figure}[!ht]
	\begin{center}
		\begin{tikzpicture}
		\begin{axis}[
		legend entries = {$U_{c}^{0}(t)$, $U_{c}^{1}(t)(n_{d}^{c}(t)=10\%)$, $U_{c}^{1}(t)(n_{d}^{c}(t)=50\%)$, $U_{c}^{1}(t)(n_{d}^{c}(t)=90\%)$, $U_{c}^{1}(t)(n_{d}^{c}(t)=100\%)$ },
		legend to name = named,
		legend columns = 3,
		xlabel = Number of users,
		ylabel = Utility of CSP,
		title = {$EF_{u}^{c}(t)= 10\% \text{ of } SF_{u}^{c}(t)$}
		]
		\addplot coordinates {(10,0.64)(20,1.29)(30,1.94)(40,2.59)(50,3.24)(60,3.89)(70,4.54)(80,5.19)(90,5.84)(100,6.49)};
		
		\addplot coordinates {(10,0.66)(20,1.26)(30,1.86)(40,2.46)(50,3.06)(60,3.67)(70,4.27)(80,4.87)(90,5.47)(100,6.07)};
		
		\addplot coordinates {(10,0.46)(20,0.87)(30,1.27)(40,1.68)(50,2.08)(60,2.49)(70,2.89)(80,3.3)(90,3.7)(100,4.11)};
		
		\addplot coordinates {(10,0.27)(20,0.48)(30,0.69)(40,0.9)(50,1.1)(60,1.31)(70,1.52)(80,1.73)(90,1.94)(100,2.15)};
		
		\addplot coordinates {(10,0.22)(20,0.38)(30,0.54)(40,0.7)(50,0.86)(60,1.02)(70,1.18)(80,1.34)(90,1.5)(100,1.66)};
		\end{axis}
		\end{tikzpicture}
		\begin{tikzpicture}
		\begin{axis}[
		xlabel = Number of users,
		ylabel = Utility of CSP,
		title = {$EF_{u}^{c}(t)= 20\% \text{ of } SF_{u}^{c}(t)$},
		]
		\addplot coordinates {(10,0.64)(20,1.29)(30,1.94)(40,2.59)(50,3.24)(60,3.89)(70,4.54)(80,5.19)(90,5.84)(100,6.49)};
		
		\addplot coordinates {(10,0.68)(20,1.3)(30,1.91)(40,2.53)(50,3.15)(60,3.77)(70,4.39)(80,5)(90,5.62)(100,6.24)};
		
		\addplot coordinates {(10,0.55)(20,1.04)(30,1.53)(40,2.02)(50,2.51)(60,3)(70,3.49)(80,3.98)(90,4.47)(100,4.96)};
		
		\addplot coordinates {(10,0.42)(20,0.78)(30,1.15)(40,1.51)(50,1.87)(60,2.23)(70,2.59)(80,2.96)(90,3.32)(100,3.68)};
		
		\addplot coordinates {(10,0.39)(20,0.72)(30,1.05)(40,1.38)(50,1.71)(60,2.04)(70,2.37)(80,2.7)(90,3.03)(100,3.36)};
		\end{axis}
		\end{tikzpicture}
		\begin{tikzpicture}
		\begin{axis}[
		xlabel = Number of users,
		ylabel = Utility of CSP,
		title = {$EF_{u}^{c}(t)= 30\% \text{ of } SF_{u}^{c}(t)$},
		]
		\addplot coordinates {(10,0.64)(20,1.29)(30,1.94)(40,2.59)(50,3.24)(60,3.89)(70,4.54)(80,5.19)(90,5.84)(100,6.49)};
		
		\addplot coordinates {(10,0.69)(20,1.33)(30,1.96)(40,2.6)(50,3.23)(60,3.86)(70,4.5)(80,5.13)(90,5.77)(100,6.4)};
		
		\addplot coordinates {(10,0.63)(20,1.2)(30,1.77)(40,2.34)(50,2.91)(60,3.48)(70,4.05)(80,4.62)(90,5.19)(100,5.76)};
		
		\addplot coordinates {(10,0.57)(20,1.07)(30,1.58)(40,2.08)(50,2.59)(60,3.1)(70,3.6)(80,4.11)(90,4.61)(100,5.12)};
		
		\addplot coordinates {(10,0.55)(20,1.04)(30,1.53)(40,2.02)(50,2.51)(60,3)(70,3.49)(80,3.98)(90,4.47)(100,4.96)};
		\end{axis}
		\end{tikzpicture}
		\begin{tikzpicture}
		\begin{axis}[
		xlabel = Number of users,
		ylabel = Utility of CSP,
		title = {$EF_{u}^{c}(t)= 40\% \text{ of } SF_{u}^{c}(t)$},
		]
		\addplot coordinates {(10,0.64)(20,1.29)(30,1.94)(40,2.59)(50,3.24)(60,3.89)(70,4.54)(80,5.19)(90,5.84)(100,6.49)};
		
		\addplot coordinates {(10,0.71)(20,1.36)(30,2.01)(40,2.66)(50,3.31)(60,3.97)(70,4.62)(80,5.27)(90,5.92)(100,6.57)};
		
		\addplot coordinates {(10,0.71)(20,1.37)(30,2.02)(40,2.68)(50,3.33)(60,3.99)(70,4.64)(80,5.3)(90,5.95)(100,6.61)};
		
		\addplot coordinates {(10,0.72)(20,1.38)(30,2.04)(40,2.7)(50,3.35)(60,4.01)(70,4.67)(80,5.33)(90,5.99)(100,6.65)};
		
		\addplot coordinates {(10,0.72)(20,1.38)(30,2.04)(40,2.7)(50,3.36)(60,4.02)(70,4.68)(80,5.34)(90,6)(100,6.66)};
		\end{axis}
		\end{tikzpicture}
		\begin{tikzpicture}
		\begin{axis}[
		xlabel = Number of users,
		ylabel = Utility of CSP,
		title = {$EF_{u}^{c}(t)= 50\% \text{ of } SF_{u}^{c}(t)$},
		]
		\addplot coordinates {(10,0.64)(20,1.29)(30,1.94)(40,2.59)(50,3.24)(60,3.89)(70,4.54)(80,5.19)(90,5.84)(100,6.49)};
		
		\addplot coordinates {(10,0.73)(20,1.39)(30,2.06)(40,2.73)(50,3.39)(60,4.06)(70,4.73)(80,5.4)(90,6.06)(100,6.73)};
		
		\addplot coordinates {(10,0.71)(20,1.53)(30,2.26)(40,3)(50,3.73)(60,4.47)(70,5.2)(80,5.94)(90,6.67)(100,7.41)};
		
		\addplot coordinates {(10,0.86)(20,1.67)(30,2.47)(40,3.27)(50,4.07)(60,4.88)(70,5.68)(80,6.48)(90,7.29)(100,8.09)};
		
		\addplot coordinates {(10,0.88)(20,1.7)(30,2.52)(40,3.34)(50,4.16)(60,4.98)(70,5.8)(80,6.62)(90,7.44)(100,8.26)};
		\end{axis}
		\end{tikzpicture}
		\\ \ref{named}
	\end{center}
	\caption{The effect of $EF_{u}^{c}(t)$ and $n_{d}^{c}(t)$ on utility of the CSP}
	\label{fig:graphsCSPUtility}
\end{figure}
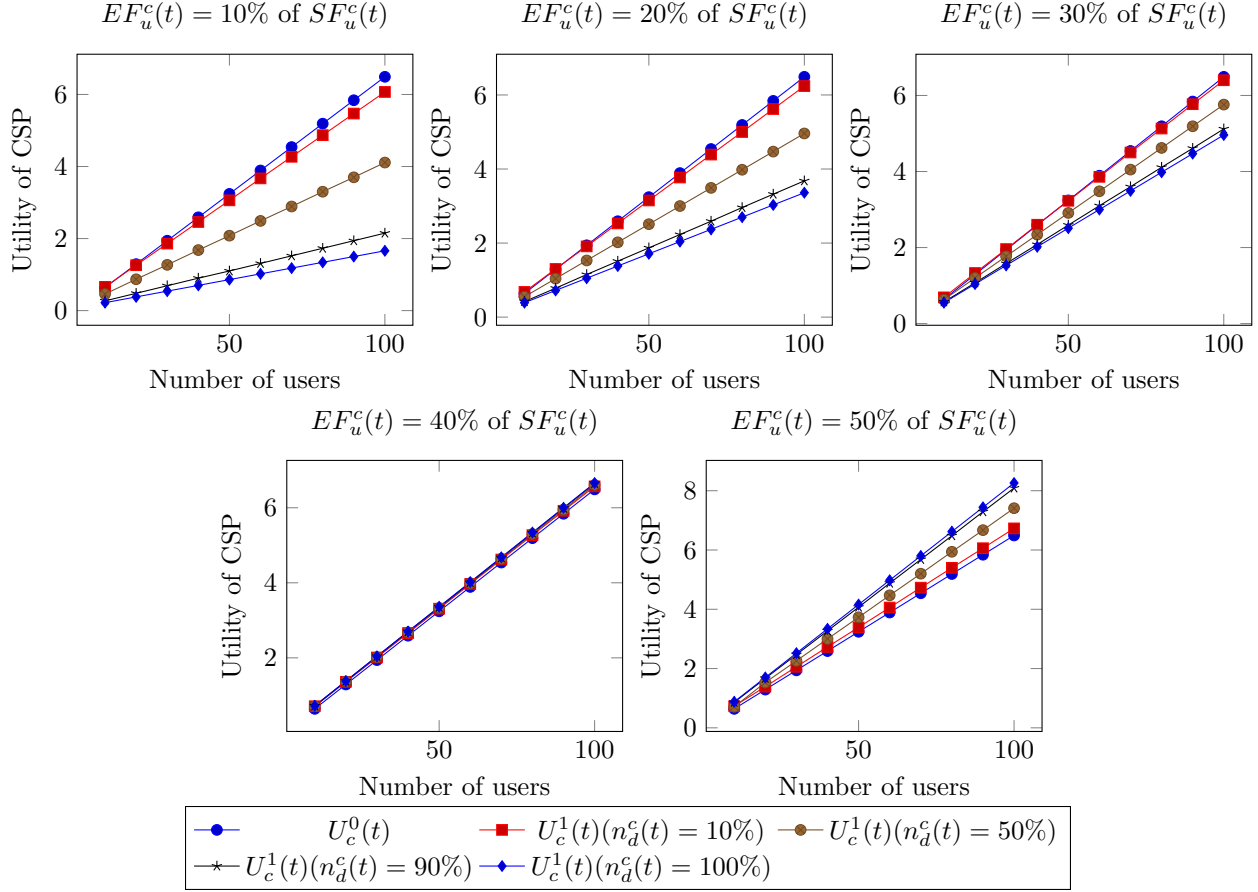
\subsection{Experiment 2: Testing $\mathcal{B}_{DEDU}$ and $\mathcal{B}_{I\text{-}DEDU}$ with public dataset}
In Experiment 1, we have considered only similar requests where all users have the same data file with the same size. Now, in this experiment, we have chosen a dataset which consists of information of Debian packages gathered from Debian popularity contest (https://popcon.debian.org/contrib/by\_inst). We took a snapshot of the number of packages, the number of installations of a single package, and size of each package on 7th May 2020. Each package represents a data file to be stored in the cloud, and the installations serve as the deduplication requests. They were a total of 403 packages (data files) and 270738 installations requests (users). The dataset is very diverse as it consists of data files having different sizes and each data file have a different number of installations. We uniformly distributed the requests among five CSPs and computed their utilities. As we discussed earlier, the inter-CSP deduplication is same as single CSP deduplication from the user point of view, and hence no change in the utility of users. But, the utility of the CSP changes as $U_{c}^{2}(t)=U_{c}^{1}(t)+AF_{in}^{c}(t)-AF_{out}^{c}(t)$. $AF^{c}(t)$ is the fee paid by a CSP $c_{i}$ to CSP $c_{j}$ for accessing the data stored at $c_{j}$. In Figure \ref{fig:interCSPUtilities}, we show the utilities of the above considered dataset by taking $EF_{u}^{c}(t)=40\%$ of $SF_{u}^{c}(t)$ and $n_{d}^{c}(t)=100\% \text{ of } N_{d}^{c}(t)$. Observe that the CSPs gain more utility when opted for inter-CSP deduplication. The gain is due to the increase in the dedup rate. Therefore, with our proposed incentive mechanism dedup with $\mathcal{B}_{I\text{-}DEDU}$ is more profitable than dedup with $\mathcal{B}_{DEDU}$ and dedup with $\mathcal{B}_{DEDU}$ is more profitable than no deduplication.
\pgfplotsset{width=0.8\textwidth,compat=1.3}
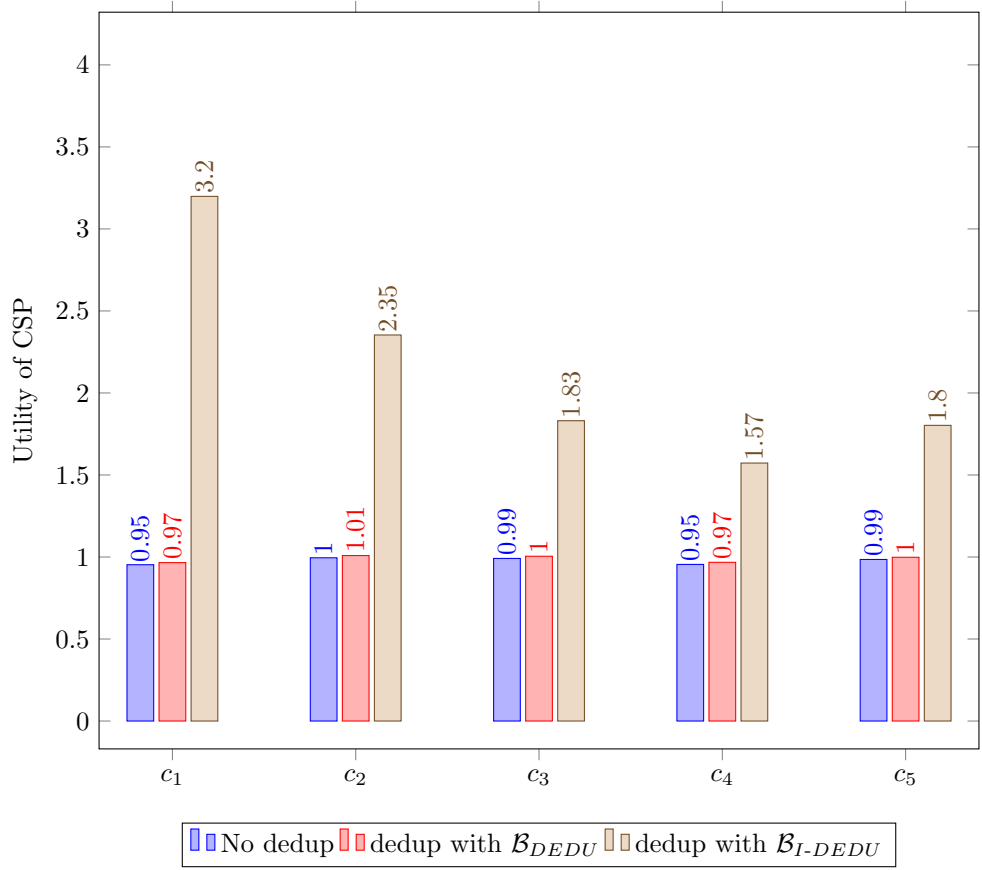
\begin{figure}[!ht]
	\centering
\begin{tikzpicture}
\begin{axis}[
x tick label style={
	/pgf/number format/1000 sep=},
ylabel=Utility of CSP,
symbolic x coords = {$c_{1}$,$c_{2}$,$c_{3}$,$c_{4}$,$c_{5}$},
xtick = data,
enlarge y limits=0.5,
legend style={at={(0.5,-0.1)},
	anchor=north,legend columns=-1},
ybar, 
nodes near coords,
 nodes near coords style={anchor=west,rotate=90,inner xsep=1pt},
]
\addplot 
coordinates {($c_{1}$,0.9529) ($c_{2}$,0.9952) ($c_{3}$,0.9913) ($c_{4}$,0.9548) ($c_{5}$,0.985)};
\addplot 
coordinates {($c_{1}$,0.9657) ($c_{2}$,1.0087) ($c_{3}$,1.0047) ($c_{4}$,0.9677) ($c_{5}$,0.9983)};
\addplot 
coordinates {($c_{1}$,3.1985) ($c_{2}$,2.3535) ($c_{3}$,1.8308) ($c_{4}$,1.5727) ($c_{5}$,1.8025)};
\legend{No dedup, dedup with $\mathcal{B}_{DEDU}$, dedup with $\mathcal{B}_{I\text{-}DEDU}$}
\end{axis}
\end{tikzpicture}
\caption{Utility of CSPs with public dataset}
\label{fig:interCSPUtilities}
\end{figure}
\begin{table}
	\centering
	\begin{tabular}{p{1.2cm}|p{1.7cm}|p{1cm}|p{2cm}|p{2.2cm}|p{0.8cm}}
		\hline
		\textbf{Scheme} & \textbf{Category}& \textbf{Incen-tive} & \textbf{Feature of Incentives} & \textbf{Correctness of deduplication rate} & \textbf{Fair-ness} \\ \hline \hline
		\cite{miao2015payment} & server-controlled& \checkmark & IR & \checkmark (Trusted party) & x\\ \hline
		\cite{li2018deduplication} & Blockchain-controlled& x & x & x & x \\ \hline
		\cite{liang2019game} & server-controlled& \checkmark & IR, IC& \text{x} & x\\ \hline
		\cite{wang2019blockchain} & client-controlled & \checkmark & - & x & partial \\ \hline
		\cite{huang2022blockchain} & server-controlled & \checkmark & - & x  & \checkmark (Arbi-tration)  \\ \hline
		\cite{ming2022blockchain} & Blockchain-controlled & x & x & x & x \\ \hline
		Proposed scheme & Blockchain-controlled& \checkmark & IR, IC& \checkmark (Blockchain)& \checkmark\\
		\hline
	\end{tabular}
	\caption{Comparison with existing works}
	\label{tab:comparisons}
\end{table}
\section{Comparison with existing schemes}
Table \ref{tab:comparisons} shows the comparison of our work with \cite{miao2015payment}, \cite{li2018deduplication}, \cite{liang2019game}, \cite{wang2019blockchain}, \cite{huang2022blockchain} and \cite{ming2022blockchain} in terms of category, incentives, features of incentives, correctness in computation of deduplication rate and fairness. 
\cite{miao2015payment} provides correctness of dedup rate but uses the services of a trusted party known as deduplication rate manager. In contrast, our scheme does not rely on the trusted party and still obtain correctness. The incentive mechanism in \cite{miao2015payment} supports only IR-constraint, whereas our scheme supports both IR-constraint and IC-constraint. The dedup scheme proposed in \cite{li2018deduplication} focuses more on the integrity of the deduplicated data stored in the cloud, whereas our scheme focuses on incentives and fair payment mechanism. The incentive mechanism in \cite{liang2019game} supports both IR-constraint and IC-constraint but, an untrusted CSP computes the dedup rate. Also, the incentives in \cite{liang2019game} are time-variant, whereas our scheme supports uniform payments. 
Although \cite{wang2019blockchain} have considered fair payments, their scheme does not support fair payments in all the cases defined in Theorem \ref{th:fairpayments}. Huang et al. provides incentives, uniform payments as well as fairness through arbitration protocol. However, the arbitration protocol is not in-built into the protocol and requires multiple rounds of offline communication. Ming et al. \cite{ming2022blockchain} emphasize more on storing and retrieving deduplication information from smart contract but not on incentives and fairness.

\section{Conclusion and future work}
Deduplication techniques save storage costs of a cloud storage provider. However, adoption of deduplication techniques by cloud users requires strong incentives and fair payment platform. In this paper, our contributions are two-fold: first, we have designed a new incentive mechanism, and second, we designed a Blockchain-based deduplication scheme. Experiment results show that the proposed incentive mechanism is individually rational and incentive compatible for both CSP and users. The proposed dedup scheme solved the problem of correctness, uniform payments and fair payments in deduplication of cloud data without a trusted intermediary. The designed smart contracts in the proposed scheme are implemented in Ethereum network, and the costs of interacting with the smart contract are presented. The limitation of our scheme is that a malicious CSP may disable the file link after receiving the fee from the contract. In the future, we would like to address this limitation by using deferred payments. Our scheme could serve as a basic model to construct more robust dedup schemes using smart contracts with less number of rounds.

 \bibliographystyle{elsarticle-num}
\bibliography{main}
\end{document}